\documentclass[a4paper,twocolumn,11pt,accepted=2023-06-05]{quantumarticle}
\pdfoutput=1

\usepackage[utf8]{inputenc}
\usepackage[english]{babel} 
\usepackage[T1]{fontenc} 

\usepackage[normalem]{ulem}
\usepackage[table,usenames,dvipsnames,x11names]{xcolor}

\usepackage[numbers]{natbib}
\usepackage[
			]{hyperref}
\PassOptionsToPackage{linktocpage}{hyperref} 

\usepackage{xcolor}
	\definecolor{quantumviolet}{HTML}{53257F} 

\providecommand{\pra}{Phys.\ Rev.\ A}

\providecommand{\prl}{Phys.\ Rev.\ Lett.}

\hypersetup{
  colorlinks   = true, 
  urlcolor     = magenta, 
  linkcolor    = quantumviolet, 
  citecolor    = black 
}
\usepackage{bookmark} 

\usepackage{graphicx}
\usepackage[caption=false, justification=justified]{subfig}
\graphicspath{{./Numerics/plots/}{./}}

\usepackage{amsthm} 
\usepackage{amssymb}
\usepackage{amsmath,mathtools}
\usepackage{dsfont}
\usepackage[vcentermath]{youngtab}
\usepackage{comment}
\usepackage{paralist}

\usepackage{algorithm}
 \usepackage{algpseudocode}
\usepackage{tikz}

\usetikzlibrary{
shadows.blur
  }

\definecolor{newblue}{RGB}{40,210,251}
\definecolor{lightgray}{RGB}{170,170,170}
\definecolor{darkyellow}{RGB}{255,210,70}
\definecolor{darkyellow2}{RGB}{251,184,38}
\definecolor{metalblue}{RGB}{78,156,219}
\definecolor{metalblue2}{RGB}{34,52,103}
\definecolor{pink}{RGB}{237,16,118}
\definecolor{pink2}{RGB}{131,28,71}
\definecolor{violet}{HTML}{53257F} 
\definecolor{violet2}{RGB}{61,18,100}
\definecolor{applegreen}{rgb}{0.55, 0.71, 0.0}
\definecolor{applegreen2}{rgb}{0.4, 0.8, 0.0}

\definecolor{DarkGray}{gray}{0.25} 
\definecolor{MidGray}{gray}{0.38} 
\definecolor{NeutralGray}{gray}{0.5}
\definecolor{LightGray}{gray}{0.8}
\definecolor{DarkRed}{rgb}{0.7,0,0}
\definecolor{DarkBlue}{rgb}{0,0,0.5}
\definecolor{SteelBlue}{rgb}{0,0.4,0.6}
\definecolor{Orange}{rgb}{0.7,0.5,0}
\definecolor{Violette}{rgb}{0.5,0,0.5}
\definecolor{Sand}{rgb}{0.84,0.8,0.55}
\definecolor{niceblue}{rgb}{0.33,0.5,0.8}
\definecolor{OliveGreen}{RGB}{0,102,102}
\definecolor{NiceGreen}{RGB}{0,153,72}

\colorlet{tensorcol}{LightGray}
\colorlet{tensorcolborder}{DarkGray}

\definecolor{changecol}{rgb}{0.7,0,0}

\colorlet{tensorcol1}{applegreen}
\colorlet{tensorcol1border}{DarkGray}
\colorlet{tensorcol2}{metalblue}
\colorlet{tensorcol2border}{metalblue2}

\newtheorem{theorem}{Theorem}

\newtheorem{lemma}[theorem]{Lemma}
\newtheorem{proposition}[theorem]{Proposition}
\newtheorem{definition}[theorem]{Definition}

\newtheorem{problem}[theorem]{Problem}
\newtheorem*{problem*}{Problem}

\newcommand{\myleft}{\mathopen{}\mathclose\bgroup\left}
\newcommand{\myright}{\aftergroup\egroup\right}
\newcommand{\e}{\ensuremath\mathrm{e}}
\renewcommand{\i}{\ensuremath\mathrm{i}}

\DeclareMathOperator{\Tr}{Tr}
\renewcommand{\Re}{\operatorname{Re}}

\DeclareMathOperator{\rank}{rank}

\DeclareMathOperator{\Id}{Id}
\DeclareMathOperator*{\argmin}{arg\,min}
\DeclareMathOperator{\supp}{supp}
\DeclareMathOperator{\diag}{diag}

\newcommand{\fro}{\mathrm{F}}

\renewcommand{\fro}{F}

\DeclareMathOperator\tr{Tr}

\makeatletter 
\newcommand\complexity@possiblymakesmaller[1]{#1} 
\newcommand\complexity@fontcommand{\mathsf} 
\newcommand{\ComplexityFont}[1]{%
{\ensuremath{\complexity@possiblymakesmaller{\complexity@fontcommand{#1}}}}
}
\makeatother
\newcommand{\NP}{\ComplexityFont{NP}}
\newcommand{\Poly}{\ComplexityFont{P}}

\newcommand{\CC}{\mathbb{C}}
\newcommand{\RR}{\mathbb{R}}

\newcommand{\NN}{\mathbb{N}}

\newcommand{\EE}{\mathbb{E}}
\newcommand{\PP}{\mathbb{P}}

\newcommand{\mc}[1]{\mathcal{#1}}

\renewcommand{\vec}[1]{\mathbf{#1}}
\newcommand{\ad}{^\dagger}

\newcommand{\norm}[1]{\left\Vert #1 \right\Vert} 
\newcommand{\snorm}[1]{\norm{#1}_\infty} 

\newcommand{\fnorm}[1]{\norm{#1}_\fro} 

\newcommand{\lTwoNorm}[1]{\norm{#1}_{\ell_2}} 

\newcommand{\ket}[1]{\left.\left|{#1}\right.\right\rangle}

\newcommand{\bra}[1]{\left.\left\langle{#1}\right.\right|}

\newcommand{\ketbra}[2]{\ket{#1} \!\! \bra{#2}}

\renewcommand{\Pr}{\operatorname{\PP}} 

\newcommand{\maximize}{\mathrm{maximize }}

\newcommand{\mat}[1]{#1}

\usepackage[vcentermath]{youngtab}
\newcommand{\scrsym}{{\Yboxdim{4pt}\,\yng(2)}}

\makeatletter 
\hypersetup{
	pdftitle = {Semi-device-dependent blind quantum tomography},
     pdfauthor = {Ingo Roth, Jadwiga Wilkens, Dominik Hangleiter, Jens Eisert},
     pdfsubject = {Compressed sensing, quantum physics},
     pdfkeywords = {Compressive sensing, compressed sensing, sensing, semi, device, dependent, independent, 
     quantum, state, tomography, blind,
     sparsity, complexity, self-calibrating, 
     low-rank, matrix, recovery, guarantee, 
     Pauli, measurement, demixing, de-mixing, hardness, projection, 
    }
}
\makeatother

\newcommand{\hzb}{Helmholtz-Zentrum Berlin f\"ur Materialien und Energie, 
Germany}
\newcommand{\fu}{Dahlem Center for Complex Quantum Systems, Freie Universit\"{a}t Berlin, Germany}

\newcommand{\abu}{Quantum Research Center, Technology Innovation Institute (TII), Abu Dhabi, UAE}
\newcommand{\quics}{Joint Center for Quantum Information and Computer Science (QuICS), University of Maryland/NIST,
USA}

\begin{document}

\title{Semi-device-dependent blind quantum tomography}

\author{Ingo Roth}
\email[Corresponding author: ]{ingo.roth@tii.ae}
\affiliation{\abu}
\affiliation{\fu}

\author{Jadwiga Wilkens}
\affiliation{\abu}
\affiliation{\fu}

\author{Dominik Hangleiter}
\affiliation{\quics}
\affiliation{\fu}

\author{Jens Eisert}
\affiliation{\fu}
\affiliation{\hzb}

\hypersetup{bookmarksdepth=-2} 

\begin{abstract}
Extracting tomographic information about quantum states is a crucial task in the quest towards devising high-precision quantum devices. Current schemes typically require measurement devices for tomography that are a priori calibrated to 
high precision. Ironically, the accuracy of the measurement calibration is fundamentally limited by the accuracy of state preparation, establishing a vicious cycle. Here, we prove that this cycle can be broken and the dependence on the measurement device's calibration significantly relaxed. We show that exploiting the natural low-rank structure of quantum states of interest suffices to arrive at a highly scalable `blind' tomography scheme with a classically efficient post-processing algorithm. We further improve the efficiency of our scheme by making use of the sparse structure of the calibrations. 
This is achieved by relaxing the blind quantum tomography problem to 
the de-mixing of a sparse sum of low-rank matrices. 
We prove that the proposed algorithm recovers a low-rank quantum state and the calibration provided that the measurement model exhibits a restricted isometry property.
For generic measurements, %
we show that 
it 
requires a close-to-optimal number of measurement settings.
Complementing these 
conceptual and mathematical insights, we numerically demonstrate that robust blind quantum tomography is possible 
in a practical setting inspired by an implementation of trapped ions. 
\end{abstract}


\maketitle

\setcounter{page}{1}
\hypersetup{bookmarksdepth}
\tableofcontents

\section{Introduction}
The development of quantum technologies is arguably one of the most vivid scientific endeavours of current times. 
This development is faced with a daunting challenge: To achieve the promising advantages of those technologies one must engineer individual quantum components with an enormous precision. The main limiting factors 
in implementing the many existing theoretical proposals for exciting applications in quantum computing today are the achievable noise levels and scalability of the components. 

From an engineering perspective, improving such \emph{noisy intermediate scale quantum (NISQ) devices} 
\cite{Preskill:2018,Roadmap} requires flexible \emph{diagnostic techniques} 
to extract actionable advice on how to improve the device in the engineering cycle. 
Such diagnostic schemes must meet tight practical constraints in terms of their complexity as well as the required accuracy of the used devices.
One of the most basic diagnostic tasks is the extraction of tomographic information about quantum states from experimentally measured data. 
Indeed, at the heart of every quantum computation is the preparation of a quantum state. 
Quantum state tomography can therefore provide valuable information for improving quantum devices beyond a mere benchmarking of their correct functioning \cite{EisertEtAl:2019}. 

However, in any such endeavour one encounters the following fundamental challenge: In order to arrive at an accurate state estimate, most tomography schemes rely on measurement devices that are calibrated to a very high precision. 
At the same time, a precise and detailed characterization of a measurement device requires an accurate state preparation. 
But improving the accuracy of the state preparation using tomographic information was our goal to begin with.
We are trapped in a vicious cycle.
This vicious cycle, depicted in Figure~\ref{fig:viciouscycle}, constitutes a fundamental obstacle to the improvement of quantum devices. 

{
Using various assumptions and models that are motivated by the specific underlying physical platform, quantum devices are routinely calibrated in a bottom-up fashion, building trust in the individual components such as steps of ground state preparation, read-out, and individual gate pulses one at a time. However, such methods are ultimately limited by the vicious cycle. 
State-of-the-art quantum computing experiments in addition employ the feedback from self-consistent
characterization methods that are robust to \emph{state-preparation and measurement} (SPAM) errors such as (linear) cross-entropy benchmarking \cite{BoiIsaSme16,SupremacyReview}, other variants of randomized benchmarking \cite{EmersonAlickiZyczkowski:2005,KniLeiRei08,MagesanGambettaEmerson:2011,HelsenEtAl:2020:General} or gate-set tomography {\cite{MerGamSmo13,BluKinNie13}} to refine the device calibration heuristically.  
Notably, these approaches are closely tied to modelling a quantum \emph{computing} device in terms of structured gate sets in addition to state preparation and measurement, and moreover {require performing sequences of multiple gates}. 
Such models are distinct from our abstraction in terms of an unknown state-preparation and an uncalibrated device. 
}

\begin{figure}
\centering
	\includegraphics[width=.5\textwidth]{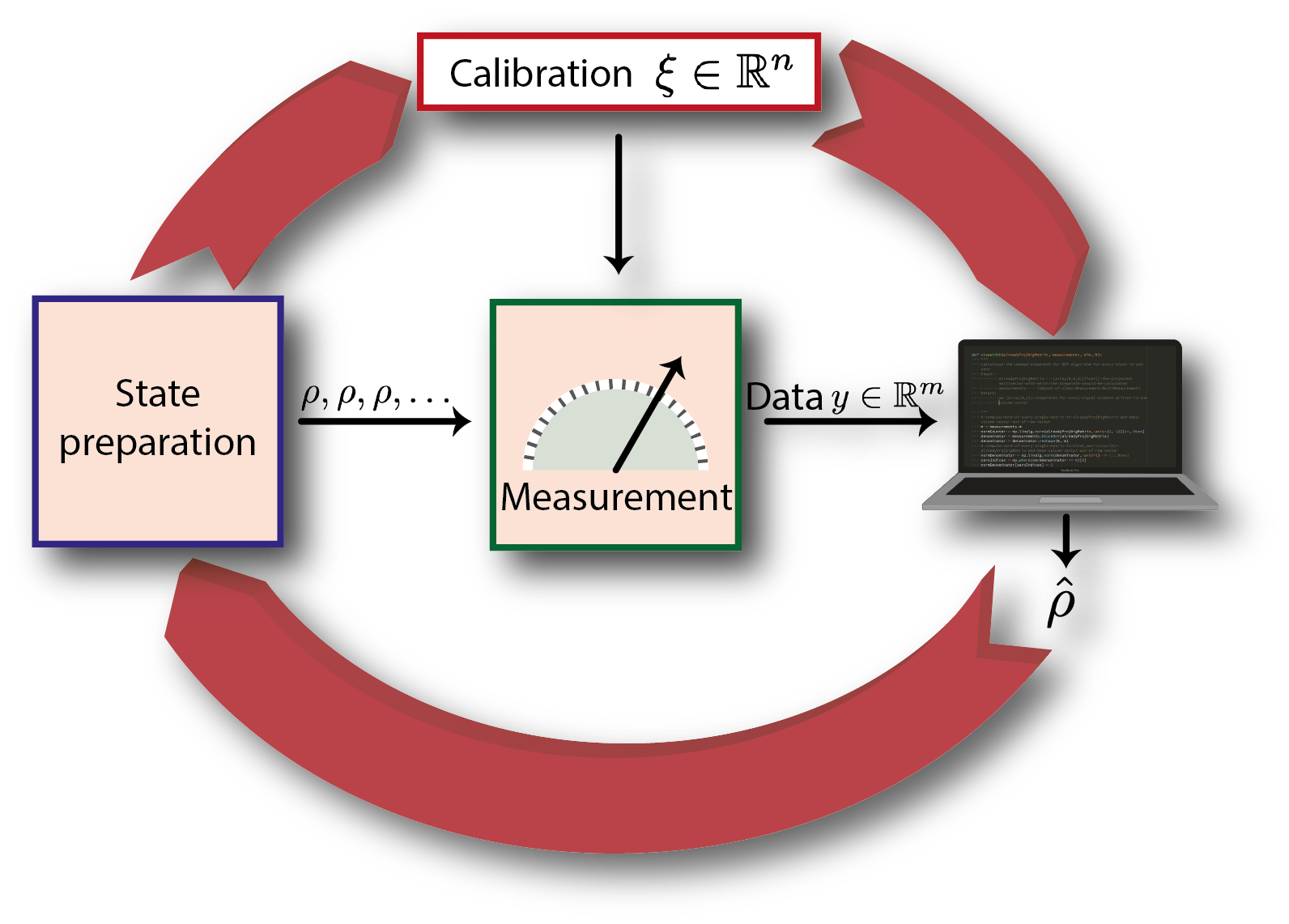}
	\vspace{-1cm}
	\caption{
	\label{fig:viciouscycle} 
	In the quest to engineer high fidelity quantum technologies one encounters a vicious cycle: 
	Extracting \emph{actionable advice} to correct for error in the state preparation requires accurate quantum state estimation. 
	The accuracy of a state estimate crucially relies on the precise calibration of the measurement device. 
	But the calibration can ultimately only be tested and improved if high fidelity quantum states are provided. 
	 }
\end{figure}

{
An important conceptual question with immediate practical relevance is therefore:}
Is there any hope to {directly} break the vicious cycle? 
In other words, can one perform quantum state tomography \emph{blindly}, that is, without full knowledge of the measurement to begin with? 
More specifically, can one simultaneously infer a quantum state and learn certain unknown calibration parameters of the measurement device in a \emph{self-calibrating tomography scheme}~\cite{BranczykEtAl:2012}? 
A simple parameter count indicates that this is typically impossible by just measuring a set of mutually orthogonal observables: 
While an arbitrary quantum state in a $d$-dimensional Hilbert space is characterized by $d^2 -1$ many real parameters, at the same time, 
the number of linearly independent measurements in this space implies that we can learn at most $d^2$ independent parameters. 
This leaves room for a single additional (calibration) parameter only and
prohibits even slightly relaxing the requirement of a complete and accurate characterization of the measurement device towards a partially uncalibrated device.
Tomography of an arbitrary quantum state is therefore typically intrinsically measurement \emph{device-dependent} in this sense. 

In this work, we break this vicious cycle. 
We observe that the above naive parameter count misses the 
point that in reasonably controlled quantum devices,
commonly encountered quantum states exhibit a natural structure: 
they are close to being pure. 
We leverage this natural property to prove that one can simultaneously learn an unknown 
calibration and a low-rank quantum state. We thus arrive at what we coin a \emph{semi-device-dependent} scheme in which the dependence on the measurement apparatus is significantly softened.

In order to achieve this goal, we formulate the \emph{blind tomography problem} as the recovery task of a highly structured signal.
This allows us to exploit and further develop a powerful formal machinery from modern signal processing 
to propose 
a scalable self-calibrating state tomography scheme that comes with mathematically guarantees. 
We use a general model of the measurement device which applies in a variety of relevant experiments:  
we model the measurements as depending linearly on the unknown parameters of the possible calibration errors. 
Indeed, in many situations the daunting uncertainty about the device calibration is small and can be approximated as a linear deviation from an empirically known calibration baseline. 

Our scheme makes a trade-off between the dependence on the measurement device and the state preparation device explicit and allow to optimally exploit this dependence in a practical scheme. 
It is an intriguing feature of our results that while structural assumptions on the quantum state to be learned typically allow for \emph{more efficient} solutions~\cite{GrossLiuFlammia:2010,FlammiaGrossLiu:2012,KalevKosutDeutsch:2015,RiofrioEtAl:2017,SteffensEtAl:2017}, here, structural assumptions allows one to solve a task in settings where it \emph{could not be solved at all} in the absence of this assumption.

Going further, we exploit yet another structure to significantly extend the realm of applicability and efficiency of our scheme, namely the \emph{sparsity} of the calibration. 
Physically, this structural property amounts to the assumption that only a small number out of the many possible calibration errors has occurred in the specific experiment. 
In our scheme we therefore simultaneously exploit the low-rank structure of the quantum state and sparsity of the calibration coefficients to 
overcome the vicious tomography cycle and provide rigorous guarantees with a favourable scaling in terms of both the system dimension and the number of calibration errors. 

\subsection{Provable blind tomography via sparse de-mixing}

Let us be slightly more formal in order to give an overview over 
{methods used, and technical contributions made in this work.}
In mathematical terms, the blind tomography task that we solve is to infer a 
 vector $\xi$ of $n$ 
calibration parameters and a rank-$r$ quantum state $\rho$ from data of the form  
\begin{equation}\label{eq:intro_data_blindtomography}
	y = \mc B_\xi(\rho) = \mc A(\xi \otimes \rho) 
\end{equation}
where $\mc B: \xi, \rho \mapsto \mc B_\xi(\rho)$ is a bi-linear map describing the measurement model. 
The measured data $y$ might for example be estimates for the expectation values of observables or probabilities of POVM elements. 
For the time being, we ignore the error of the estimates induced by finite statistics. 
It is convenient to regard the data as associated to a structured linear estimation problem: 
we can equivalently model the measurement map as a linear map $\mathcal A$ acting on $\xi \otimes \rho$. 

Such structured linear inverse problems are studied in the mathematical discipline of model-based 
\emph{compressed sensing} \cite{BaraniukCevherDuarteHegde:2010, FoucartRauhut:2013}, where efficient algorithms with analytical performance guarantees have been developed.
A work horse of compressed sensing that most rapidly solve the relevant inverse problems are 
so-called \emph{iterative hard-thresholding (IHT)} 
algorithms \cite{BlumensathDavies:2008}. 
{In this work, we will use this general algorithmic paradigm, study the novel thresholding operations that arise in our context and prove new recovery guarantees.}

As a first result of this work, we establish that the key step of an IHT algorithm that solves the 
blind tomography problem is \NP{}-hard. 
To overcome this obstacle, we propose an IHT algorithm that solves a slightly relaxed version of the blind tomography problem: 
the task of de-mixing a
 sum of $n$ different low-rank quantum states $\rho_i$, i.e.,\ data of the form
\begin{equation}\label{eq:intro_data_sparsede-mixing}
	y = \mathcal{A}\left(\sum_{i=1}^n \xi_i e_i \otimes \rho_i\right),
\end{equation}
where $\{e_i\}_{i=1}^n$ denotes the standard orthonormal basis. 
An efficient IHT algorithm for the de-mixing problem of low-rank matrices was developed and analysed in Ref.~\cite{StrohmerWei:2017}. 
This algorithm can be readily adapted to our problem. 

But relaxing the blind tomography problem to the de-mixing problem artificially introduces an overhead in the number of unknown degrees of freedom of the problem scaling as $2drn$, and in particular linearly with the number of calibration parameters in the model. 
This leads to an unfavourable situation in a two-fold manner:
First, determining many calibration parameters also requires many measurement settings as the cost per calibration parameter scales with the dimension $d$ of the quantum system. 
Second, a necessary condition for a well-posed blind de-mixing problem of rank-$r$ with a maximal number of $d^2$ linearly independent measurements of the form \eqref{eq:intro_data_sparsede-mixing} is that there are more linearly independent measurements than real parameters, i.e., $2 r n d \leq d^2$. 
This means that the simultaneous determination of a certain number of calibration parameters $n$ can in principle only work for sufficiently large system dimension $d$ in many situations.
This causes severe constraints in the achievable self-calibration for small system sizes. 

We argue that an additional well-motivated structural assumption can render the blind tomography much more broadly applicable. 
This structure is exploited in our new {\emph{sparse-demixing thresholding} (SDT) algorithm.} 
Our argument is based on the observation that the problem of determining an accurate estimate of the quantum state in the blind setting involves solving two distinct sub-problems: 
first, one needs to determine which ones of many potential error models of the measurement contribute. 
Second, one needs to estimate the calibration parameters of these models.
Generically, there are many potential models that parametrize, for instance, the deviation of every imperfect implementation of a fixed measurement setting from its ideal implementation. 

In this case, the first problem becomes combinatorially costly since many distinct measurement settings need to be simultaneously calibrated.
In contrast, in our approach, it is straightforward to solve both tasks simultaneously and even avoid a combinatorial overhead using the built-in relaxations of compressed sensing. 
To this end, we observe that allowing for many potential errors with associated calibration parameters only a small number $s$  of which contribute amounts to assuming that the calibration vector $\xi$ is $s$-sparse, i.e.,\ it has only $s$ non-vanishing entries. 
Of course, we do not assume that we know the support of the vector $\xi$. 
This falls naturally into the framework of structured signal recovery. 
To summarize: we observe data generated by linear measurements acting on $\xi \otimes \rho$ where $\xi$ is an $s$-sparse vector and $\rho$ is a rank $r$ quantum state. 

We are now faced with the recovery problem of de-mixing a \emph{sparse sum} of different \emph{low-rank quantum states}. 
We show that the projection onto this structure can be efficiently calculated using hierarchical thresholding \cite{RothEtAl:2016} and therefore circumvents our \NP{}-hardness result. 
We derive the corresponding iterative hard-thresholding algorithm and prove that it successfully recovers the states $\rho_i$ and the sparse vector $\xi$ provided that the measurement map $\mc A$ acts isometrically on sparse sums of low-rank states. 
We further show that generic measurement ensembles with $m$ different measurement settings exhibit this restricted isometry property provided that $m$ scales at least as $srd+s\log n$.
Thus, we find that our algorithm solves the blind tomography problem with an overhead in the required number of measurements that scales linearly in $s$  as compared to the number of degrees of freedom in the problem given by $rd + s$. 
In particular, the number of potential calibration models $n$ enters only logarithmically in the measurement complexity of the scheme. 
This renders the scheme highly scalable in $n$ providing flexibility in the modelling of systematic measurement errors or calibration corrections.
Furthermore, it leaves sufficiently many linearly independent parameters to allow one to infer a couple of calibration parameters already for comparably small system sizes.
We demonstrate the performance of the algorithm for the physically relevant case of measuring Pauli operators that are locally mixed with the unknown calibration parameters. 
{Our results do not only answer a practically-inspired, conceptional question at hand in the context of blind quantum tomography, but at the same time contribute to the mathematical framework of compressed sensing as such and have potential application in other engineering disciplines.}

\subsection{Practical blind tomography}

Going beyond working out the theoretical guarantees, we numerically demonstrate the functioning of the scheme and the mindset behind it. 
Specifically, we show that the iterative hard-thresholding algorithm solves the blind tomography problem from much fewer samples than competing methods from generic (Gaussian) measurements as well as sub-sampled random Pauli measurements. 
We then take the theoretical model to the practical testbed and turn to a realistic model of measurement errors given by a coherent over-rotation along some axis. 
Those measurements have significantly more structure. 
We observe that the measurement structure together with the sparsity constraints causes the SDT algorithm to frequently 
get stuck at objective variables with an incorrect support. 
For this reason, we also study the performance of a more pragmatically minded optimization strategy, namely,
constrained alternating minimization that does not require the relaxation to the de-mixing problem. 
We numerically demonstrate that the blind tomography problem in a realistic setting can be solved using this adapted algorithmic approach. 
Thereby, we show that exploiting the low-rank structures of quantum states allows for performing tomography blindly in realistic calibration and measurement models. 
{These findings may serve as a strong motivation and invitation to 
translate our approach to 
a variety of concrete experimental settings that are practically relevant in the quantum technologies. 
The main theoretical prerequisite is to identify plausible linear (or simple non-linear) calibration models for concrete measurement implementations.  
However, it is a formidable experimental task in itself to identify an application and demonstrate a concrete  advantage of performing state tomography blindly compared to using standard bottom-up protocols, other semi-device-dependent robust calibration approaches or error mitigation techniques. 
}

\subsection{Related work and applications in signal processing}
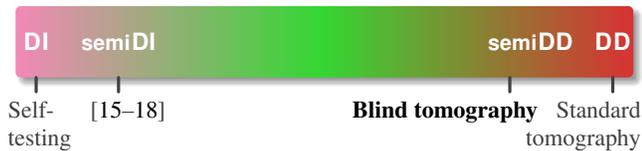
\begin{figure}
\centering
	\begin{tikzpicture}[scale = .95]
		\path [rounded corners=2pt, 
			draw =  none,
			left color = pink!50,
			right color = Red3!80,
			middle color = Green3!80,
			blur shadow={shadow blur steps=20}
			] (0,0) rectangle  (1.01\linewidth,1);
		\def\DDlabel{D\kern .5pt D}
		\def\DIlabel{D\kern .5pt I}
		\def\sDDlabel{{\footnotesize s\kern .1pt e\kern .1pt m\kern .1pt i}\kern .5pt D\kern .1pt D} 
		\def\sDIlabel{{\footnotesize s\kern .1pt e\kern .1pt m\kern .1pt i}\kern .5pt D\kern .1pt I}
		
		\def\di{\linewidth/29};
		\def\sdi{\linewidth/5};
		\def\sdd{4 * \linewidth/5};
		\def\dd{28 * \linewidth/29};
		\def\bt{3.74 * \linewidth/5}
		\def\sdim{\linewidth/5}
		\node[color=white, font=\bfseries\sffamily] at (\di,.5) {\DIlabel}; 
		\node[color=white, font=\bfseries\sffamily] at (\dd,.5) {\DDlabel}; 
		\node[color=white, font=\bfseries\sffamily] at (\sdi,.5) {\sDIlabel}; 
		\node[color=white, font=\bfseries\sffamily] at (\sdd,.5) {\sDDlabel}; 

		\draw [DarkGray, very thick, line cap = round] (\di,0) -- node [at end,below, anchor = north west, xshift = -.5cm] {\begin{minipage}{.2\linewidth}
		\begin{flushleft}
			\footnotesize \sf Self-\\testing
		\end{flushleft}
		\end{minipage}
		}(\di, -.2) ;

		\draw [DarkGray, very thick, line cap = round] (\sdim,0) -- node [at end,below, anchor = north west,  xshift = -.4cm] {
		\begin{minipage}{.25\linewidth}
		\begin{flushleft}
		\footnotesize \sf \cite{pawlowski_semi-device-independent_2011, liang_semi-device-independent_2011, li_semi-device-independent_2011, li_semi-device-independent_2012}
		\end{flushleft}
		\end{minipage}
		}(\sdim, -.2) ;
		\draw [DarkGray, very thick, line cap = round] (\bt,0) -- node [at end,below,black,anchor = north east, xshift = .4cm] {
		\begin{minipage}{.3\linewidth}
		\begin{flushright}
		\footnotesize \sf	 \textbf{Blind tomography}
		\end{flushright}
		\end{minipage}}(\bt, -.2) ;

		\draw [DarkGray, very thick, line cap = round] (\dd,0) -- node [at end,below,anchor = north east,xshift = .5cm] {\begin{minipage}{.2\linewidth}
		\begin{flushright}
			\footnotesize \sf Standard \\tomography
		\end{flushright}
		\end{minipage}}(\dd, -.2) ;

	\end{tikzpicture}
	\caption{
	\label{fig:semiDD} 
	Illustration of the spectrum between fully \emph{device-independent} (DI) and fully \emph{device-dependent} 
	(DD) quantum system characterization methods such as self-testing and standard tomography, respectively. 
	\emph{Semi-device-independent} (semi-DI) methods relax the stringent requirements of full device-independence. 
	Self-calibrating tomography relaxes the assumptions on the calibration of the measurement 
	device and therefore exemplifies a \emph{semi-device-dependent} (semi-DD) scheme. 
	The blind tomography scheme presented here is an example of such a semi-device-dependent scheme.
	 }
\end{figure}

In our semi-device-dependent, self-calibrating scheme we aim at softening the requirements of fully device-dependent schemes that crucially rely on perfect measurement apparata. 
Coming from the opposite end, in quantum communications introducing mild assumptions such as bounds on the system dimension \cite{gallego_device-independent_2010}, one can weaken the impractical stringency of full device independence to semi-device independence~\cite{pawlowski_semi-device-independent_2011, liang_semi-device-independent_2011, li_semi-device-independent_2011, li_semi-device-independent_2012}.
Device-independent and device-dependent approaches can, thus, be seen as the extreme ends of an axis that quantifies the amount of assumptions on the measurement device, 
see Figure \ref{fig:semiDD} for an illustration. Hand in hand with reducing the amount of assumptions and gaining robustness to imperfections, the amount of novel information that can be gained is dramatically reduced. 
Semi-device-independent schemes move away from the requirements of full device-independence towards more practical settings but are still extremely demanding in terms of the required resources and acceptance criteria.
Our semi-device-dependent tomography scheme lies in the opposite regime. It slightly relaxes the assumptions on the precision of the measurement apparatus but still extracts tomographic information.

Self-calibrating tomography schemes have been previously proposed in specific contexts using different methods and assumptions as a leverage to break the vicious blind tomography cycle. 
In Ref.~\cite{Mogilevtsev:2010} it has been argued that single photon detectors can be simultaneously calibrated during state tomography under the assumptions that the state is squeezed extending the mindset of Ref.~\cite{MogilevtsevEtAl:2009}; see also Ref.~\cite{MogilevtsevEtAl:2012} for a more extensive discussion of potential classes of states and the recent Ref.~\cite{sim_proper_2019} for error bars in this context. 
Ref.~\cite{BranczykEtAl:2012} has reported experimental demonstration of simultaneously reconstructing a quantum state together with certain unknown unitary rotations associated to the measurement device via maximum likelihood estimation in a linear optics setting. Here, the term `self-calibrating tomography' was coined. 
From a practical perspective, this work is perhaps closest in mindset to the current work. 
Complementing and going significantly beyond this work, here we prove that such a self-calibrating approach works under very mild and natural structural assumptions and give rigorous guarantees. 

Another general-purpose framework is the \emph{Gram matrix completion} 
proposed in Refs.~\cite{Stark:2012, Stark:2014}. 
Here, a correlation matrix encoding information about the measurement, the state and the measured data is completed from a subset of known indices. 
So-called data-pattern tomography \cite{rehacek_operational_2010} avoids the calibration of the measurement device by comparing the data to previously determined signatures of well-controlled reference states such as coherent states \cite{motka_efficient_2014}.
More conceptually speaking, schemes incorporating model selection for quantum state tomography can be also viewed as self-calibrating~\cite{ferrie_quantum_2014}.

Another set of approaches focuses on characterizing entire gate sets or their respective unitary errors self-consistently from the observed statistics when applying different sequences \cite{MerGamSmo13, BluKinNie13, Gre15, BluGamNie17, cerfontaine_self-consistent_2019,BriegerRothKliesch:2023:cGST}.
These methods typically rely on a certain design of the measurement sequences and cost-intensive classical post-processing. 
In contrast, our model of the measurement devices as linearly depending on a set of calibration parameters 
is much simpler and requires far less resources.

Our work builds-on and further develops compressed-sensing techniques for the tomography of quantum devices. 
Previous compressed-sensing schemes for quantum tomography reduce the effort in the data acquisition while still ensuring an efficient classical post-processing 
\cite{GrossLiuFlammia:2010,Gross:2011,Liu:2011,KalevKosutDeutsch:2015,Kueng:2015,KabanavaEtAl:2016}. 
These schemes come with theoretical guarantees and have as well been successfully employed in experiments \cite{SteffensEtAl:2017,RiofrioEtAl:2017,ShabaniEtAl:2011}. 
The practical applicability of compressed sensing tomography schemes rests on their robustness and stability against various imperfections of the experimental setup. 
Small deviations from the compressive model assumption and additive errors to the measurement outcomes, 
e.g.\ induced by finite statistics, 
are reflected in a proportional and only slightly enhanced estimation error.
Still, the schemes rely on measurement devices that are calibrated to very high precision with the notable exception of compressive tomography schemes for quantum processes that use randomized benchmarking data \cite{KimLiu15, RothEtAl:2018}. 
Here, we relax this requirement using a semi-device-dependent approach. 
In distinction, in the previous schemes low-rank assumptions were considered to reduce the complexity of a tomography scheme, giving 
rise to an important \emph{quantitative improvement}. 
Here, those assumptions are expected to often make blind tomography possible in the first place and therefore permit even a \emph{qualitative improvement} over the known schemes.

Recovery problems of the form \eqref{eq:intro_data_blindtomography} or the related de-mixing problem \eqref{eq:intro_data_sparsede-mixing} also arise in other disciplines. 
For example, these problems appear in future mobile communication scenarios with the promise to yield much more scalable protocols with respect to the number of served devices \cite{WunderBocheStrohmerJung:2015}. 
More specifically, our work can be applied in order to extend the \emph{internet-of-things setup} described in Ref.~\cite{StrohmerWei:2017} in case one additionally wants to exploit the sporadic (sparse) user activity of machine-type messaging.
Furthermore, our work identifies yet another set of hierarchical signal structures that allow for an efficient projection: 
It extends the work on compressed sensing with hierarchically sparse signals of a subset of the authors to low-rank matrices \cite{RothEtAl:iTwist:2016,RothEtAl:2016}.

The remainder of this work is organized as follows. 
In the subsequent Section~\ref{sec:application}, we give a detailed description of a concrete experimental setup 
that motivates our mathematical formulation of the blind tomography problem. 
In Section~\ref{sec:setup}, we provide the formal definitions of the blind tomography problem and introduce the notation used in the subsequent parts of the work. The details of the sparse demixing algorithm and its variant based on alternating optimization are derived in Section~\ref{sec:algorithm}. 
On the way, we establish the \NP{}-hardness of the projection associated to the original blind tomography problem. 
The theorems guaranteeing the performance of the sparse demixing algorithm are explained in Section~\ref{sec:guarantees}. 
The corresponding proofs are given in the appendix. 
Finally, numerical simulations of the algorithms performance 
and its application to practical use cases
are shown in Section~\ref{sec:numerics} before we conclude with an outlook in Section~\ref{sec:outlook}.  

\section{Quantum state tomography with imperfect Pauli correlation measurements}\label{sec:application}

So far, our description of the measurement scheme has been fairly abstract. 
In the following, we describe a concrete scenario in which our formalism applies. 
Consider an ion trap experiment preparing a multi-qubit quantum state $\rho$. 
We perform Pauli correlation measurements, i.e., we estimate $m$ expectation values of the form 
\begin{equation}\label{eq:targetSO}
	\mathcal A_0(\rho)^{(k)} = \tr\left[\rho \left(W^{(k)}_1 \otimes W^{(k)}_2 \otimes \cdots \otimes W^{(k)}_l\right)\right],
\end{equation} 
where $W^{(k)}_j \in \{X, Y, Z, \Id\}$ is a Pauli matrix acting on the $j$th qubit and $k \in [m] \coloneqq \{1, 2, \ldots, m\}$. 
We refer to $\mathcal A_0: \CC^{d\times d} \to \RR^m$ as the measurement map or sampling operator~\footnote{
Note that one might actually implement projective measurements in the multi-qubit Pauli basis as done, e.g., in Refs.~\cite{RiofrioEtAl:2017,SteffensEtAl:2017}. 
While such projective measurements contain more information than the Pauli correlation measurement, we restrict ourselves to Pauli expectation values both for the sake of simplicity and to remain in a setting for which theoretical guarantees can be proven \cite{GrossLiuFlammia:2010,Liu:2011}.}. 

In many experimental setups, it is natural to implement measurements of a certain Pauli observable -- in the case of ion traps Pauli $Z$ -- while the other Pauli observables require more effort. 
A measurement of any other Pauli observable -- in the case of ion traps Pauli $X$ and Pauli $Y$ -- can then be implemented by applying a suitable sequence of unitary gates prior to the measurement. 
For example, using addressed laser pulses of different duration one can implement rotations around different axes 
and thus implement the Hadamard gate $H$ as well as the phase gate $S$.
In this way, one can realize measurements in the $X = H Z H$ and $Y = S H Z H S\ad$ basis.

But each application of an additional gate may come with a coherent error \emph{in addition} to the native error associated with the measurement itself. 
In this way, we end up with different systematic errors for different Pauli observables parametrized by the angles $\theta, \varphi$ of a coherent error given by $\e^{\i \theta X}\e^{\i \varphi Z}$. 
This gives rise to some probability of actually measuring the expectation value of another local Pauli matrix than the targeted one. 
For example, consider a coherent error given by a (small) rotation around the $Z$-axis as given by $\e^{i\varphi Z}$. 
The faulty implementation of the Hadamard gate is then given by $\tilde H =\e^{i\varphi Z}  H$. 
Of course, the native $Z$-measurement is untouched by this coherent error, since no unitary rotation precedes this measurement. 
However, instead of $Y$ one now actually measures $\tilde Y = S \tilde H Z \tilde H\ad S\ad = \cos(2 \varphi) Y +  \sin (2 \varphi) X $. 
At the same time $X$ remains undisturbed.

More generally, we can introduce calibration parameters $\xi_{W \to \tilde W}$ measuring the strength of the error that replaces a certain target Pauli matrix $W$ by $\tilde W$. 
For instance, in the above example those parameters are given by $\xi_{Y \to Y} = \cos(2 \varphi)$, $\xi_{Y \to X} = \sin(2 \varphi)$ and $\xi_{Z \to Z} = \xi_{X \to X} = 1$. 
For simplicity, we assume that these calibration parameters are identical for different qubit registers.
Assuming that errors are not too large, the calibration parameters for the target measurement fulfil $\xi_{W \to W} \approx 1$ 
or all $W \in \{X, Y, Z\}$.
This leaves us with six independent calibration parameters corresponding to the cross-contributions. 
To construct the measurement map $\mathcal A{}$, we start from the definition of the target measurement $\mathcal A_0$ in \eqref{eq:targetSO}.
From $\mathcal A_0$ we can derive calibration measurement components $\mathcal A_{W\to \tilde W}$ appearing with the coefficient $\xi_{W \to \tilde W}$ by replacing all appearances of the Pauli matrix $W$ in the definition of $\mathcal A_0$ with $\tilde W$. 
If $W$ appears in a multi-qubit Pauli observable several times the resulting observable is the sum of all Pauli observables generated by replacing only one of the $W$ by $\tilde W$, assuming that the coherent errors are small so that the higher-order terms can be neglected.
For example, a faulty realization of the observable $ZYZZY$ is now given by $ \xi_{Y \rightarrow Y} ZYZZY + \xi_{Y \to X} (ZXZZY + ZYZZX)$. 

Altogether, to linear order in the calibration parameters $\xi_{W \to \tilde W}$ with $W \neq \tilde W$ we end up constructing a description of the effective faulty measurement by 
\begin{equation}
	\label{eq:paulireplacement}
	y = \xi_0\mathcal{A}_0(\rho) + \sum_{W \neq \tilde W \in \{X, Y, Z\}} \xi_{W \to \tilde W} \mathcal{A}_{W\to \tilde W}(\rho),  
\end{equation}
which can be written as linear map $\mathcal A$ action on $\xi \otimes \rho$ with $\xi = [\xi_0, \xi_{X\to Y}, \xi_{X \to Z}, \ldots, \xi_{Z \to Y}]^T$. By assumption, we set $\xi_0 = 1$. 

In this measurement model the sparsity assumption is justified if unitary errors in a certain coordinate plane are dominant compared to others thus singling out certain types of calibration measurement components. 
Importantly, we do not assume that we know which corrections are dominant (i.e.,  the support of $\xi$) \emph{a priori}. 
The measurement model also exemplifies a setting in which one is ultimately limited to measuring 
a maximal set of $d^2$ observables. Thus, blind tomography becomes only possible exploiting structure assumptions 
if one does not allow for different ways of implementing the same measurement that yield different calibration corrections without introducing too many new calibration parameters. 

\section{Formal problem definition}\label{sec:setup}
Motivated by this example, we set out to provide a formal definition of the blind tomography problem and the related sparse-de-mixing problem. 
The notation and terminology introduced in this section allows us to formulate a general signal-processing framework using which the blind tomography can be provably solved.
Both, the blind tomography and the sparse de-mixing, problems are linear inverse problems that feature a combination of different compressive structures. 
These are smaller sets of linear vector spaces, and it will be convenient to introduce some notation to refer to these sets. 
The prototypical example is the \emph{set of $s$-sparse real vectors}
\begin{equation*}
	\Sigma_s^n \coloneqq \{\xi \in \RR^n \mid |\supp\xi| \leq s \} \subset \RR^n,
\end{equation*}
which is defined by the support $\supp\xi$ of a vector $\xi$, i.e.,\ the index set of the non-vanishing entries of $\xi$, having cardinality smaller or equal than $s$. 
The set of $s$-sparse vectors is not a vector space itself but the union of $\binom n s$ $s$-dimensional subspaces.

In the realm of quantum mechanics, the non-commutative analogue of sparse vectors, namely low-rank matrices, is 
important. 
We denote the \emph{set of complex rank $r$ matrices} by
\begin{equation*}
	\CC^{d\times d}_r \coloneqq \{ x \in \CC^{d\times d} \mid \rank x \leq r\}.
\end{equation*}
Since we are dealing with quantum states we will restrict our attention to the set $\mathcal{D}^d \subset \CC^{d\times d}$ of trace-normalized, positive semidefinite matrices, i.e., $\rho \geq 0$ and $\tr \rho = 1$ for all $\rho \in \mathcal D^d$.
Our results can be straightforwardly generalized to general matrices without these constraints. We denote the set of rank-$r$ quantum states as $\mathcal D^d_r = \mathcal D^d \cap \CC^{d\times d}_r$. In particular, $\mathcal D^d_1$ is the set of pure quantum states.
In order to solve the blind tomography problem we need to simultaneously recover an $s$-sparse real vector $\xi$ and a rank-$r$ quantum state $\rho$. 
It is convenient to regard  both $\xi$ and $\rho$ as a combined signal $X = \xi \otimes \rho$ 
and model the measurement including its dependence on the calibration parameter as a linear map $\mathcal A$ acting on $X$. 
Considering such linear maps instead of bi-linear maps is sometimes referred to as `lifting' in the compressed sensing literature \cite{AhmedRechtRomberg:2014}. 
For a physicist, `lifting' is also the natural isomorphism at the heart of the density matrix formulation of quantum mechanics.
The signal $X$ is highly structured as it is a tensor product of a sparse vector and a low-rank quantum state. We denote the set of all potential signals as 
\begin{equation*}
	\Omega^{n,d}_{s,r} \coloneqq \{\xi \otimes x \mid \xi \in \Sigma^n_s,\ x \in \mathcal D^d_r \} \subset \CC^{nd \times d}.
\end{equation*}
One can regard a signal $X \in \Omega^{n,d}_{s,r}$ as an $nd \times d$ matrix consisting of $n$ blocks of size $d\times d$ stacked on top of each other as depicted in 
Figure~\ref{fig:signalstructures}, where each $d\times d$ block is proportional to the same quantum state $\rho$ and only $s$ of the blocks are non-vanishing.

\begin{figure}[tb]
\begin{minipage}{.08\textwidth}
\centering
\begin{tikzpicture}[x=.25cm,y=.25cm]
	\coordinate(A) at (0,0);
      \coordinate(B) at (0,1);
      \coordinate(C) at (0,2);
      \coordinate(D) at (0,3);
      \coordinate(E) at (0,4);
      \coordinate(F) at (0,5);

      \coordinate(AA) at (1,0);
      \coordinate(BA) at (1,1);
      \coordinate(CA) at (1,2);
      \coordinate(DA) at (1,3);
      \coordinate(EA) at (1,4);
      \coordinate(FA) at (1,5);

      \draw[thick,fill=red] (A)--(B)--(BA)--(AA)--(A);
      \draw[thick, fill=blue] (B)--(C)--(CA)--(BA)--(B);
      \draw[thick,fill=orange] (C)--(D)--(DA)--(CA)--(C);
      \draw[thick,fill=violet] (D)--(E)--(EA)--(DA)--(D);
      \draw[thick, fill=magenta] (E)--(F)--(FA)--(EA)--(E);

      \coordinate(A) at (0,5);
      \coordinate(B) at (0,6);
      \coordinate(C) at (0,7);
      \coordinate(D) at (0,8);
      \coordinate(E) at (0,9);
      \coordinate(F) at (0,10);

      \coordinate(AA) at (1,5);
      \coordinate(BA) at (1,6);
      \coordinate(CA) at (1,7);
      \coordinate(DA) at (1,8);
      \coordinate(EA) at (1,9);
      \coordinate(FA) at (1,10);

      \draw[thick,fill=yellow] (A)--(B)--(BA)--(AA)--(A);
      \draw[thick, fill=Maroon] (B)--(C)--(CA)--(BA)--(B);
      \draw[thick, fill=Turquoise] (C)--(D)--(DA)--(CA)--(C);
      \draw[thick,fill=green] (D)--(E)--(EA)--(DA)--(D);
      \draw[thick, fill=purple] (E)--(F)--(FA)--(EA)--(E);

      \draw[draw=white](.5,-1.5) rectangle ++(1,1)  node[pos=.5]{$\CC^{nd \times d }$};

\end{tikzpicture}

\end{minipage}
\begin{minipage}{.15\textwidth}
\centering
  \begin{tikzpicture}[x=.25cm,y=.25cm]
  
      \coordinate(A) at (0,2.5);
      \coordinate(B) at (0,3);
      \coordinate(C) at (0,3.5);
      \coordinate(D) at (0,4);
      \coordinate(E) at (0,4.5);
      \coordinate(F) at (0,5);

      \coordinate(AA) at (0.5,2.5);
      \coordinate(BA) at (0.5,3);
      \coordinate(CA) at (0.5,3.5);
      \coordinate(DA) at (0.5,4);
      \coordinate(EA) at (0.5,4.5);
      \coordinate(FA) at (0.5,5);

      \draw[thick,fill=black!50] (A)--(B)--(BA)--(AA)--(A);
      \draw[thick] (B)--(C)--(CA)--(BA)--(B);
      \draw[thick,fill=black!50] (C)--(D)--(DA)--(CA)--(C);
      \draw[thick,fill=black!50] (D)--(E)--(EA)--(DA)--(D);
      \draw[thick] (E)--(F)--(FA)--(EA)--(E);

      \coordinate(A) at (0.0,5);
      \coordinate(B) at (0.0,5.5);
      \coordinate(C) at (0.0,6);
      \coordinate(D) at (0.0,6.5);
      \coordinate(E) at (0.0,7);
      \coordinate(F) at (0.0,7.5);

      \coordinate(AA) at (0.5,5);
      \coordinate(BA) at (0.5,5.5);
      \coordinate(CA) at (0.5,6);
      \coordinate(DA) at (0.5,6.5);
      \coordinate(EA) at (0.5,7);
      \coordinate(FA) at (0.5,7.5);

      \draw[thick,fill=black!50] (A)--(B)--(BA)--(AA)--(A);
      \draw[thick] (B)--(C)--(CA)--(BA)--(B);
      \draw[thick] (C)--(D)--(DA)--(CA)--(C);
      \draw[thick,fill=black!50] (D)--(E)--(EA)--(DA)--(D);
      \draw[thick] (E)--(F)--(FA)--(EA)--(E);
      \draw [draw=black, thick] (1.6,5) arc (0:360:0.3);
      \draw [thick] (1.1,4.79)--(1.5,5.2);
      \draw [thick] (1.1,5.2)--(1.5,4.79);
      \draw[fill=red,thick] (2,4.5)--(3,4.5)--(3,5.5)--(2,5.5)--(2,4.5);
      \node[text width=.01cm, thick] at (3.2,5) {\tiny=};

      \coordinate(A) at (4.2,0);
      \coordinate(B) at (4.2,1);
      \coordinate(C) at (4.2,2);
      \coordinate(D) at (4.2,3);
      \coordinate(E) at (4.2,4);
      \coordinate(F) at (4.2,5);

      \coordinate(AA) at (5.2,0);
      \coordinate(BA) at (5.2,1);
      \coordinate(CA) at (5.2,2);
      \coordinate(DA) at (5.2,3);
      \coordinate(EA) at (5.2,4);
      \coordinate(FA) at (5.2,5);

      \draw[thick,fill=red] (A)--(B)--(BA)--(AA)--(A);
      \draw[thick] (B)--(C)--(CA)--(BA)--(B);
      \draw[thick,fill=red] (C)--(D)--(DA)--(CA)--(C);
      \draw[thick,fill=red] (D)--(E)--(EA)--(DA)--(D);
      \draw[thick] (E)--(F)--(FA)--(EA)--(E);

      \coordinate(A) at (4.2,5);
      \coordinate(B) at (4.2,6);
      \coordinate(C) at (4.2,7);
      \coordinate(D) at (4.2,8);
      \coordinate(E) at (4.2,9);
      \coordinate(F) at (4.2,10);

      \coordinate(AA) at (5.2,5);
      \coordinate(BA) at (5.2,6);
      \coordinate(CA) at (5.2,7);
      \coordinate(DA) at (5.2,8);
      \coordinate(EA) at (5.2,9);
      \coordinate(FA) at (5.2,10);

      \draw[thick,fill=red] (A)--(B)--(BA)--(AA)--(A);
      \draw[thick,] (B)--(C)--(CA)--(BA)--(B);
      \draw[thick] (C)--(D)--(DA)--(CA)--(C);
      \draw[thick,fill=red] (D)--(E)--(EA)--(DA)--(D);
      \draw[thick] (E)--(F)--(FA)--(EA)--(E);

      \draw[draw=white](1.5,-1.5) rectangle ++(1,1)  node[pos=.5]{$\Omega^{n,d}_{s,r}$};

\end{tikzpicture}
\end{minipage}
\begin{minipage}{.22\textwidth}
\flushright
\centering
\begin{tikzpicture}[x=.25cm,y=.25cm]
    
      \coordinate(A) at (4,2.5);
      \coordinate(B) at (4,3);
      \coordinate(C) at (4,3.5);
      \coordinate(D) at (4,4);
      \coordinate(E) at (4,4.5);
      \coordinate(F) at (4,5);

      \coordinate(AA) at (4.5,2.5);
      \coordinate(BA) at (4.5,3);
      \coordinate(CA) at (4.5,3.5);
      \coordinate(DA) at (4.5,4);
      \coordinate(EA) at (4.5,4.5);
      \coordinate(FA) at (4.5,5);

      \draw[thick] (A)--(B)--(BA)--(AA)--(A);
      \draw[thick] (B)--(C)--(CA)--(BA)--(B);
      \draw[thick] (C)--(D)--(DA)--(CA)--(C);
      \draw[thick] (D)--(E)--(EA)--(DA)--(D);
      \draw[thick] (E)--(F)--(FA)--(EA)--(E);

      \coordinate(A) at (4,5);
      \coordinate(B) at (4,5.5);
      \coordinate(C) at (4,6);
      \coordinate(D) at (4,6.5);
      \coordinate(E) at (4,7);
      \coordinate(F) at (4,7.5);

      \coordinate(AA) at (4.5,5);
      \coordinate(BA) at (4.5,5.5);
      \coordinate(CA) at (4.5,6);
      \coordinate(DA) at (4.5,6.5);
      \coordinate(EA) at (4.5,7);
      \coordinate(FA) at (4.5,7.5);

      \draw[thick] (A)--(B)--(BA)--(AA)--(A);
      \draw[thick] (B)--(C)--(CA)--(BA)--(B);
      \draw[thick] (C)--(D)--(DA)--(CA)--(C);
      \draw[thick,fill=black!50] (D)--(E)--(EA)--(DA)--(D);
      \draw[thick] (E)--(F)--(FA)--(EA)--(E);

      \draw [draw=black, thick] (5.6,5) arc (0:360:0.3);
      \draw [thick] (5.1,4.79)--(5.5,5.2);
      \draw [thick] (5.1,5.2)--(5.5,4.79);
      \draw[fill=green,thick] (6,4.5)--(7,4.5)--(7,5.5)--(6,5.5)--(6,4.5);

      \node[thick] at (7.6,5) {\footnotesize+};
      \node[thick] at (9,5) {\scriptsize$\cdots$};
      \node[thick] at (10.1,5) {\footnotesize+};

      \coordinate(A) at (11,2.5);
      \coordinate(B) at (11,3);
      \coordinate(C) at (11,3.5);
      \coordinate(D) at (11,4);
      \coordinate(E) at (11,4.5);
      \coordinate(F) at (11,5);

      \coordinate(AA) at (11.5,2.5);
      \coordinate(BA) at (11.5,3);
      \coordinate(CA) at (11.5,3.5);
      \coordinate(DA) at (11.5,4);
      \coordinate(EA) at (11.5,4.5);
      \coordinate(FA) at (11.5,5);

      \draw[thick,fill=black!50] (A)--(B)--(BA)--(AA)--(A);
      \draw[thick] (B)--(C)--(CA)--(BA)--(B);
      \draw[thick] (C)--(D)--(DA)--(CA)--(C);
      \draw[thick] (D)--(E)--(EA)--(DA)--(D);
      \draw[thick] (E)--(F)--(FA)--(EA)--(E);

      \coordinate(A) at (11,5);
      \coordinate(B) at (11,5.5);
      \coordinate(C) at (11,6);
      \coordinate(D) at (11,6.5);
      \coordinate(E) at (11,7);
      \coordinate(F) at (11,7.5);

      \coordinate(AA) at (11.5,5);
      \coordinate(BA) at (11.5,5.5);
      \coordinate(CA) at (11.5,6);
      \coordinate(DA) at (11.5,6.5);
      \coordinate(EA) at (11.5,7);
      \coordinate(FA) at (11.5,7.5);

      \draw[thick] (A)--(B)--(BA)--(AA)--(A);
      \draw[thick] (B)--(C)--(CA)--(BA)--(B);
      \draw[thick] (C)--(D)--(DA)--(CA)--(C);
      \draw[thick] (D)--(E)--(EA)--(DA)--(D);
      \draw[thick] (E)--(F)--(FA)--(EA)--(E);
      \draw [draw=black, thick] (12.6,5) arc (0:360:0.3);
      \draw [thick] (12.1,4.79)--(12.5,5.2);
      \draw [thick] (12.1,5.2)--(12.5,4.79);

      \draw[fill=red,thick] (13,4.5)--(14,4.5)--(14,5.5)--(13,5.5)--(13,4.5);
      \node[thick] at (14.8,5) {\footnotesize=};

      \coordinate(A) at (15.5,0);
      \coordinate(B) at (15.5,1);
      \coordinate(C) at (15.5,2);
      \coordinate(D) at (15.5,3);
      \coordinate(E) at (15.5,4);
      \coordinate(F) at (15.5,5);

      \coordinate(AA) at (16.5,0);
      \coordinate(BA) at (16.5,1);
      \coordinate(CA) at (16.5,2);
      \coordinate(DA) at (16.5,3);
      \coordinate(EA) at (16.5,4);
      \coordinate(FA) at (16.5,5);

      \draw[thick,fill=red] (A)--(B)--(BA)--(AA)--(A);
      \draw[thick] (B)--(C)--(CA)--(BA)--(B);
      \draw[thick,fill=orange] (C)--(D)--(DA)--(CA)--(C);
      \draw[thick,fill=violet] (D)--(E)--(EA)--(DA)--(D);
      \draw[thick] (E)--(F)--(FA)--(EA)--(E);

      \coordinate(A) at (15.5,5);
      \coordinate(B) at (15.5,6);
      \coordinate(C) at (15.5,7);
      \coordinate(D) at (15.5,8);
      \coordinate(E) at (15.5,9);
      \coordinate(F) at (15.5,10);

      \coordinate(AA) at (16.5,5);
      \coordinate(BA) at (16.5,6);
      \coordinate(CA) at (16.5,7);
      \coordinate(DA) at (16.5,8);
      \coordinate(EA) at (16.5,9);
      \coordinate(FA) at (16.5,10);

      \draw[thick,fill=blue] (A)--(B)--(BA)--(AA)--(A);
      \draw[thick,] (B)--(C)--(CA)--(BA)--(B);
      \draw[thick] (C)--(D)--(DA)--(CA)--(C);
      \draw[thick,fill=green] (D)--(E)--(EA)--(DA)--(D);
      \draw[thick] (E)--(F)--(FA)--(EA)--(E);

      \draw[draw=white](10,-1.5) rectangle ++(1,1)  node[pos=.5]{$\hat\Omega^{n,d}_{s,r}$};
    \end{tikzpicture}
    \end{minipage}

	\caption{
	\label{fig:signalstructures}
	The signal sets of the blind tomography and sparse de-mixing problem can be regarded as subsets of $\CC^{nd \times d}$, i.e.,\ matrices consisting of $n$ blocks of $d\times d$. For a blind tomography signal in $\Omega^{n,d}_{s,r}$, 
	only $s$ out of the $n$ blocks are non-zero and are proportional to the same rank $r$ matrix. In contrast, a signal of the sparse de-mixing problem in $\hat\Omega^{n,d}_{s,r}$ comprises $s$ non-vanishing blocks with potentially different rank $r$ matrices. 
	 }
\end{figure}
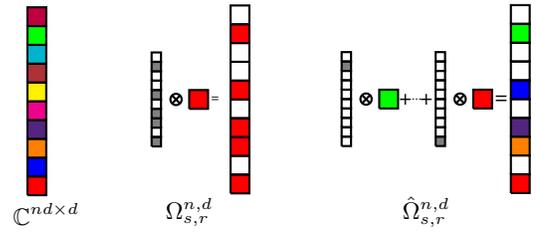

We are now equipped to concisely state the problem we would like to study. 
\begin{problem}[Blind tomography]\label{prob:blindqst} Let $\mathcal A: \CC^{nd\times d} \to \RR^m$ be a linear map. 
Given data $y = \mathcal A(X) \in \RR^m$ and the linear map $\mathcal A$, recover $X$ under the assumption that 
\begin{equation*}
	X \in \Omega^{n,d}_{s,r}.
\end{equation*}
\end{problem}
Our approach to algorithmically solving the blind tomography problem makes use of a proxy problem: 
we relax it to signals that are a bit less restrictively structured
\begin{equation*}
	\hat\Omega^{n,d}_{s,r} \coloneqq \left\{
			\sum_{i=1}^n \xi_i e_i \otimes x_i 
			\ \middle| \ 
			\xi \in \Sigma^n_s,\ 
			x_i \in \mathcal D^d_r
		\right\}.
\end{equation*}
Both sets $\hat\Omega^{n,d}_{s,r}$ and $\Omega^{n,d}_{s,r}$ are subsets of $\CC^{nd\times d}$. 
The difference between them as illustrated in Figure~\ref{fig:signalstructures} is the following: 
While for $X \in \Omega^{n,d}_{s,r}$ all $d\times d$ blocks are proportional to the same quantum state $x$, we allow the $d\times d$ blocks of $\hat X \in \Omega^{n,d}_{s,r}$ to be proportional to \emph{different quantum states} $x_i$.
Again only $s$ out of the $n$ blocks of $\hat X$ are non-vanishing. 
Analogously to Problem~\ref{prob:blindqst}, we define the linear inverse problem associated with $\hat\Omega$. 
\begin{problem}[Sparse de-mixing]\label{prob:sparsede-mixing} Let $\mathcal A: \CC^{nd\times d} \to \RR^m$ be a linear map. 
Given data $y = \mathcal A(X) \in \RR^m$ and the linear map $\mathcal A$, recover $X$ under the assumption that 
\begin{equation*}
	X \in \hat\Omega^{n,d}_{s,r}.
\end{equation*}
\end{problem}%
The observed data of the sparse de-mixing problem can be equivalently described as 
\begin{equation}\label{eq:observed_data}
	y = \sum_{k = 1}^n \xi_k \mathcal A_k(x_k),
\end{equation}
where we have split up $X$ into trace-normalized $d\times d$ blocks $x_k$ and their norm $\xi_k$ according to the definition of $\hat\Omega^{n,d}_{s,r}$. Correspondingly, we can decompose the linear map $\mathcal A$ into the set of linear maps $\{\mathcal A_k\}_{k=1}^n$ where each $\mathcal A_k$ acts only on the $k$th $d\times d$ block of $X$. 
From this reformulation it becomes clear that the problem amounts to reconstructing a set of low-rank signals $\{x_k\}_k$ from observing its sparse mixture under linear maps, hence the name sparse de-mixing.

For both the blind-tomography and the sparse-de-mixing problem, 
we alternatively write each of the $n$ linear maps $\mathcal A_k$ in terms of 
$m$ observables 
\begin{equation*}
\{A_k^{(i)} \in \CC^{d\times d} \mid (A_k^{(i)})\ad = A_k^{(i)}\}_{i=1}^m
\end{equation*}
via
\begin{equation}\label{eq:measurement_observable}
	\mathcal{A}_k(x_k)^{(i)} = \langle A_k^{(i)}, x_k \rangle 
\end{equation}
with the Hilbert-Schmidt inner product $\langle X, Y\rangle = \Tr(X\ad Y)$. 

Note that as long as we consider Hermitian matrices for the measurement $A_k^{(i)}$ and signals $x_i$, we end up with a real data vector $y \in \RR^m$. 
For applications other than quantum tomography it is straightforward to adopt our proofs and results to real signals or complex-valued measurement maps. 
Furthermore, for the sake of simplicity we have formulated both recovery problems without noise. More generally, the data can be assumed to be of the form $y = \mathcal A(X) + \epsilon$ where $\epsilon$ denotes additive, e.g.\ statistical, noise. 

In the following, we will also make use of the inner product of vectors $x, y \in \RR^n$ defined as $\langle x, y \rangle = \sum_i x_i y_i$, their $\ell_2$-norm $\|x\|_{\ell_2} = \sqrt{\langle x, x\rangle}$ and the Frobenius norm a matrix $X \in \CC^{d_1, d_2}$ induced by the Hilbert-Schmidt inner product $\|X\|_F = \sqrt{\langle X, X \rangle}$. 

\section{Algorithm}\label{sec:algorithm}

We now turn to the technical derivation of our algorithm for the blind quantum tomography and the sparse de-mixing problem. 
Our algorithm builds on primitives developed in the field of compressed sensing. 
In particular, we generalize the hard thresholding algorithm to accommodate the structural assumptions of both problems. 
As a first step, we establish the hardness of direct thresholding approaches to the blind tomography problem before stating a tractable algorithm for the sparse de-mixing problem.

Let us be more precise: the blind quantum tomography problem requires different assumptions on two levels. 
First, we want the signal to be a tensor product $\xi \otimes \rho$, i.e., of rank one. 
Second, both tensor factors are assumed to be structured. 
Concretely, we assume $\xi$ to be $s$-sparse and $\rho$ to be of rank $r$. 
We are therefore faced with low-rank structures on two separate levels: 
first, the block-structured signal as given by the tensor product of calibration vector and quantum state has unit rank. 
Second, by assumption the target quantum states, i.e.,\ the individual blocks of the signal, have low rank.

It has been observed in the compressed sensing literature that multi-level structures with structured tensor components can be notoriously hard to reconstruct. 
One prototypical example of this is combined sparsity and low-rankness in the sense that the signal is the tensor product of two sparse vectors, i.e., $X = \xi \otimes \tau$ with $\xi, \tau \in \Sigma^n_s$. 
This problem is already very similar to the blind tomography problem where one of the sparse vectors is replaced by a low-rank matrix, the quantum state.

The obstacle arising from such structures can be understood from a different perspective present in the compressed sensing literature that is related to different algorithmic approaches. 
The perhaps most prominent approach in compressed sensing is the convex relaxation of structure-promoting regularisers yielding efficient convex optimization programs. 
Minimizing the $\ell_1$-norm or the Schatten-$1$-norm is known to solve linear inversion problems involving sparse or low-rank vectors efficiently and sampling optimal, respectively. 
However, simply combining both regularisers in a convex fashion does not yield a sampling-optimal reconstruction of problems that feature both structures~\cite{OymakFazelEldarHassibi:2015}. 

\subsection{Hard-thresholding algorithms: Ease and hardness of the projection}

Another algorithmic approach used in compressed sensing are so-called \emph{hard thresholding algorithms} such as CoSAMP, IHT or HTP \cite{NeedellTropp:2008, BlumensathDavies:2008, Foucart:2011}; see also the textbook \cite{FoucartRauhut:2013} for an introduction. 
These are typically iterative procedures that minimize the deviation from the linear constraints in some way or other, e.g.\ by gradient descent, and in each iteration project onto the structure of the signal. 
For many compressed sensing problems this is possible because even though recovery problems, such as 
\begin{equation*}
	\operatorname*{minimize}_\xi \| \mathcal{A}(\xi) - y  \|_{\ell_2}^2 \quad \text{subject to $\xi \in \Sigma^n_s$}, 
\end{equation*}
are \NP{}-hard \cite{Magdon-Ismail:2017}, the related projection
\begin{equation*}
	P_{\Sigma_s^n}(\tau) \coloneqq \argmin_{\xi \in \Sigma_s^n} \| \xi - \tau \|_{\ell_2}
\end{equation*}
can be computed efficiently.
For the given example of projecting onto $s$-sparse vectors, this solution is given by the hard-thresholding operation defined as follows: 
Let $\Sigma_\text{max}$ be the set of indices of the $s$ absolutely largest entries of $\tau$. 
Then, 
\begin{equation*}
	(P_{\Sigma_s^n}(\tau))_i = \begin{cases}
		\tau_i & \text{for $i \in \Sigma_\text{max}$} \\
		0 & \text{otherwise}.
	\end{cases}
\end{equation*}
In words, one keeps the largest entries of $\tau$ and replaces the other entries by zero. 
Analogously, the projection of Hermitian matrices onto low-rank matrices can be efficiently calculated by calculating the eigenvalue decomposition and applying $P_{\Sigma_r^d}$ to the eigenvalue vector.
Let $X \in \CC^{d\times d}$ be a Hermitian matrix with eigenvalue decomposition $X = U\operatorname{diag}(\lambda) U\ad$. 
We define the projection onto positive semi-definite low-rank matrices as 
\begin{equation*}
	P_{\mathcal D_r^d}(X) = U\operatorname{diag}(P_{\Sigma_r^d}(\lambda|_{\geq 0})) U\ad,
\end{equation*}
where $\lambda|_{\geq 0}$ denotes the restriction of $\lambda$ to its non-negative entries. 

In hard-thresholding algorithms, the problems associated with simultaneously exploiting sparse and low-rank structures are manifest in the computational hardness of computing the respective projections. 
For the case of unit rank matrix with sparse singular vectors, calculating the projection is the so-called \emph{sparse PCA} problem, i.e., given a matrix $A \in \RR^{n\times n}$
\begin{equation*}
	\operatorname*{minimize}_{\xi, \tau \in \Sigma^n_s}\ \| A - \xi \otimes \tau \|_F.
\end{equation*}

Indeed, exactly solving this problem in the worst case is \NP{}-hard by a trivial 
reduction to the CLIQUE problem~\cite{Magdon-Ismail:2017}. 
But it turns out that the hardness is much worse: 
one can even make average-case hardness statements based 
on conjectures regarding the hardness of the \emph{planted clique problem}~\cite{berthet_complexity_2013,berthet_optimal_2013,brennan_optimal_2019}. 
Moreover, the SparsePCA problem remains just as hard even when one merely 
asks for an approximation up to a 
\emph{constant relative error}~\cite{chan_approximability_2016,Magdon-Ismail:2017}. 

As the first technical result of this work, 
we show that also the projection onto $\Omega^{n,d}_{s,r}$ is an \NP{}-hard problem by reducing it to the sparse PCA problem. 

\begin{theorem}[Hardness of constrained minimization]
	There exists no polynomial time algorithm that calculates
	\begin{equation*}
		\operatorname*{minimize} \quad \| A - X \|_F \quad\text{\upshape subject to $X \in \Omega^{n,d}_{s,r}$},
	\end{equation*}
	for all $A \in \CC^{nd\times d}$ unless $\Poly = \NP$. This still holds for $s=n$. 
\end{theorem}
The details of the proof are given in Appendix~\ref{app:hardnessproof}. 
This hardness result provides a strong indication that a straightforward adaptation of compressed sensing techniques is not feasible. In this work, our way out of this is to sacrifice sampling optimality of the algorithm for a lower runtime and being able to prove analytical performance guarantees. 
Alternating minimization approaches that make the factorization explicit is also a viable way forward. 
{We} provide a detailed description of such an algorithm in Section~\ref{sec:algorithm:als}. 
But proving global recovery guarantees for non-convex algorithms typically becomes much more involved. 

\subsection{Relaxing the blind tomography problem: sparse de-mixing}

In fact, the bi-sparse and low-rank structure can be relaxed to a simple hierarchical sparsity 
constraint~\cite{WunderEtAl:2018, FourcartEtAl:2019}. 
A vector $\xi \in \CC^{Nn}$ consisting of $N$ blocks of size $n$ is called $(s,\sigma)$-hierarchically sparse if it has at most $s$ blocks with non-vanishing entries, that themselves are $\sigma$-sparse \cite{SprechmannEtAl:2010,FriedmanEtAl:2010, SprechmannEtAl:2011, SimonEtAl:2013}. 
For this structure a hard-thresholding algorithm together with theoretical recovery guarantees has been 
derived in Refs.~\cite{RothEtAl:2016,RothEtAl:iTwist:2016, RothEtAl:2018:Kronecker,HierarchyIEEE}. 
It has been applied in different contexts \cite{WunderRothFritschekEisert:2017, WunderRothFritschekEisert:2018:MTC, WunderEtAl:2018:MIMO:WSA, WunderEtAl:2019:TWC} including 
\emph{sparse blind deconvolution} \cite{WunderEtAl:2018} which features the combined low-rank, sparse structure.

Here, we make use of this approach to solve the blind quantum tomography problem as formalized in Problem~\ref{prob:blindqst}. 
At the heart of our approach is the insight 
that the projection onto $\hat\Omega^{n,d}_{s,r}$ can be efficiently computed since the $n$ $d \times d$ blocks may be different. 
This allows one to combine the projection onto $\Sigma^n_s$ and the projection onto $\mathcal D^d_r$: 
First, the low-rank projection $P_{\mathcal D^d_r}$ is applied to each of the $d\times d$ blocks of the input matrix $X$. 
Subsequently, the sparse projection operator is applied by setting the $n-s$ smallest blocks in Frobenius norm to zero. 
The resulting algorithm is summarized as Algorithm~\ref{alg:projOmegaHat}, where $Y_{\overline{W}}$ denotes the subvector of $Y$ indexed by the entries in the complement of $W$. 
\begin{figure}[tb]
\begin{algorithm}[H]
	\caption{Projection onto $\hat\Omega^{n,d}_{s,r}$}\label{alg:projOmegaHat}
	\begin{algorithmic}[1]
		\Require $X \in \CC^{nd\times d}.$
		\State $Y = 0$
		\For{$k \in [n]$}
			\State $Y_k = P_{\mathcal D^d_r}(x_k)$
			\State $n_k = \|Y_k\|_{F}$
		\EndFor
		\State $W = \operatorname{supp}P_{\Sigma^n_s}(n)$
		\State $Y_{\overline W} = 0$.
		\Ensure $Y$ is projection of $X$ onto $\hat\Omega^{n,d}_{s,r}$.
	\end{algorithmic}
\end{algorithm}
\end{figure}
The computational cost of the projection onto $\hat \Omega^{n,d}_{s,r}$ is dominated by the eigenvalue decomposition required to compute the low-rank approximation $\mathcal D^d_{r}$ of each block. 
Computing the full eigenvalue decomposition of the $d\times d$ blocks requires computation time of $O(d^3)$ using,
e.g., Householder reflections \cite{Golub}. 
Since we are only interested in the dominant $r \ll d$ eigenvalues, the effort can be reduced to 
$O(rdw)$ using the Lanczos algorithm, where $w$ is the average number of non-zero elements in a row of
a block \cite{Golub}. 
Using randomized techniques one might be able to further reduce the computational costs 
\cite{HalkoEtAl:2011}.
The calculation of the Frobenius norms contributes $O(nd^2)$ flops. 
The largest blocks can be selected using the quick-select algorithm \cite{Hoare:1961} in $O(n)$. 
Note that the low-rank projections and Frobenius norms of all 
blocks can also be performed in parallel without any modification to the algorithm. 

Equipped with an efficient projection for $\hat\Omega_{s,r}^{n,d}$, 
we can construct a structured iterative gradient descent algorithm. 
This is a variant of the IHT algorithm, that was originally developed for sparse vectors 
\cite{BlumensathDavies:2008}.
The IHT algorithm is a projective gradient descent algorithm that iteratively alternates 
gradient steps to optimize the $\ell_2$-norm deviation between the data and a projection onto the constraint set. 
The resulting recovery algorithm for the \emph{sparse de-mixing (SDT)} 
problem is stated as 
Algorithm~\ref{alg:SDTalgCompact}.
\begin{figure}[tb]
\begin{algorithm}[H]
	\caption{SDT algorithm}\label{alg:SDTalgCompact}
	\begin{algorithmic}[1]
		\Require Data $y$, measurement $\mc A$, sparsity $s$ and rank $r$ of signal
		\State Initialize $X^0=0$.
		\Repeat
			\State Calculate step-widths $\mu^l$
			\State $G^l = \mathcal{A}^\dagger
							\left(
								y - \mathcal{A}(X^l)
							\right) $
			\State $X^{l+1} 
					= P_{\hat\Omega^{n,d}_{s,r}}
					\left(
						X^l + \diag(\mu^l) P_{\mathcal{T}_{X^l}}
						\left(G^l%
						\right)
					\right)$
		\Until stopping criterion is met at $l=l^\ast$
		\Ensure Recovered signal $X^{l^\ast}$
	\end{algorithmic}
\end{algorithm}
\end{figure}

The SDT algorithm is closely related to the IHT algorithm for de-mixing low-rank matrices that was developed in Ref.~\cite{StrohmerWei:2017}. 
We will refer to this algorithm as the \emph{DT algorithm}.
The main difference between our SDT and the DT algorithm of Ref.~\cite{StrohmerWei:2017} is that the latter does not make the additional sparsity assumptions on the signal. 
For this reason, the SDT algorithm differs in the projection $P_{\hat\Omega^{n,d}_{s,r}}$ that additionally applies the projection $P_{\Sigma^n_{s}}$ selecting the $s$ dominant blocks. 
In fact, in the special case of considering non-sparse signals in $\hat\Omega^{n,d}_{n,r}$ the SDT algorithm coincides with the DT algorithm. 

\subsection{Details of the SDT algorithm}

To be fully self-contained, let us now go through the individual steps of the SDT algorithm and specify the relevant details. 
Every iteration of the algorithm starts with the computation of $G^l = \mathcal A^\dagger(y - \mathcal A(X^l))$, the gradient for the $\ell_2$-norm deviation $f(X) = \frac12 \|y - \mathcal A ( X)\|_{\ell_2}^2$ evaluated at $X^l$. 
The algorithm subsequently employs a modification from Ref. \cite{WeiEtAl:2016} in calculating the steepest gradient inspired by geometrical optimization techniques which leads to a faster convergence \cite{AbsilSepulchre:2009,Vandereycken:2013}: 
The set of rank $r$ matrices is an embedded differential manifold in the linear vector space of all matrices. Thus, a direction on this embedded manifold is characterized by a tangent vector on the manifold.
While this geometry straight-forwardly generalizes to the set of $nd\times d$ matrices with rank $r$ blocks, due to sparsity constraint $\hat\Omega^{n,d}_{s,r}$ fails to be a differential manifold. 
Nonetheless, we can make use of tangent vectors as `natural' search directions in our optimization problem for the non-vanishing blocks of $X^l$ that are conforming with the fixed-rank constraint. 

The tangent space of rank $r$ matrices at point $x$ is given by the set of matrices that share the same column or row space $x$ \cite{AbsilSepulchre:2009}. 
Correspondingly, the tangent space projection of a non-vanishing block of $X$ can be defined as follows:
Let $x_k = U_k \Lambda_k U^\dagger_k$ be the eigenvalue decomposition of the $k$th block of $X$ with $\Lambda_k$ the diagonal matrix with eigenvalues in decreasing order. 
Further, let $U^{(r)}_k$ denote the restriction of $U_k$ to its first $r$ columns corresponding to the range of $x_k$. 
Then, the tangent space projection acting on $g_k$ the $k$th block of $G$ is given by 
\begin{equation}\label{eq:tangentspace_proj_block}
	P_{\mathcal{T}_{X}}(G)_k = g_k - (\Id - P_U) g_k (1 - P_U), 
\end{equation}
with $(P_U)_k = U^{(r)}_k (U^{(r)}_k)^\dagger$. The entire tangent-space projection $P_{\mathcal{T}_{X}}(G)$ is defined by acting trivially on the blocks of $G$ corresponding to vanishing blocks of $X$ and as the projection \eqref{eq:tangentspace_proj_block} otherwise.

As we prove below in generic situations the SDT algorithm converges for a constant step-width set to $\mu_l = 1$ and even without using the tangent space projection. 
Empirically, a faster convergence is achieved with the tangent space projection and using the following prescription for the step-width calculation:
From the projected gradient $G^l_P = P_{\mathcal{T}_{X}}(G^l)$ in the $l$th iteration we then calculate the algorithm's step width for each block individually as 
\begin{equation*}
	\mu_k^l = \frac
		{\|(G_P^l)_k\|_F^2}
		{\|\mathcal A ((G_P^l)_k)\|_{\ell_2}^2}
\end{equation*}
and multiply each block by the corresponding $\mu_k^l$. 
In order to have a compact notation, we introduce the diagonal matrix $\diag(\mu^l) = \diag(\mu^1_l, \ldots , \mu_1^l, \mu_2^l, \ldots, \mu_2^l, \ldots, \mu_n^l)$ where each step width is repeated $d$ times.
The new state of the algorithm, $X^{l+1}$, is given by the projection of the result of a gradient step with step width $\mu^l$ onto the set $\hat \Omega_{s,r}^{n,d}$.

Finally, we have to specify a stopping criterion at which the loop of the algorithm is exited. 
We terminate the algorithms if the objective function is below a specified threshold, i.e., 
\begin{equation}\label{eq:stoppingcriterion}
	\frac{
		\|y - \mathcal A(X^l)\|_{\ell_2} 
	}{
		\|y\|_{\ell_2}
	}
	\leq \gamma_\text{break}
\end{equation}
or a maximal number of iteration is reached. 
If the data vector $y$ has additive noise, $\gamma_\text{break}$ has to be chosen to be larger than the expected norm of the noise. 
To be less relying on expectations on the noise levels, one can alternatively make use of criteria on the gradient and step width or 
test for oscillating patterns in the identified support. 

\subsection{Blind tomography via alternating least-square optimization}\label{sec:algorithm:als}
A more direct algorithmic approach to the blind tomography problem is to use a constrained \textit{alternating least square} (ALS) optimization. 
In ALS optimization, one performs a constrained optimization of the objective function 
\begin{align*}
	f_\text{ALS}(\xi, \rho) = \frac12 \| y - \mathcal A(\xi, \rho)\|_{\ell_2}^2,
\end{align*}
 with respect to one of the two variables while regarding the respective other variable as constant in an alternating fashion, see Algorithm~\ref{alg:ALSBT}.

\begin{figure}[tb]
\begin{algorithm}[H]
	\caption{ALS-BT algorithm }\label{alg:ALSBT}
	\begin{algorithmic}[1]
		\Require Data $y$, measurement $\mc A$, sparsity $s$ and rank $r$ of signal
		\State Initialize $\rho^0$.
		\Repeat
			\State\label{alg:ALSBT:xi} %
			\vspace*{-\baselineskip}%
			\begin{equation*}
				\hspace{-7em}\quad \xi^l = 
				\argmin_{\xi \in \Sigma^n_s}\ 
				f_\text{ALS}(\xi, \rho^{l-1})
			\end{equation*}
			\State \label{alg:ALSBT:rho}%
			\vspace*{-\baselineskip}%
			\begin{equation*}
				\hspace{-7em}\rho^l = 
				\argmin_{\rho \in \CC^{d \times d}_r}\ 
				f_\text{ALS}(\xi^l, \rho)
			\end{equation*}
		\Until stopping criterion is met at $l=l^\ast$
		\Ensure Recovered signal $\rho^{l^\ast}$, $\xi^{l^\ast}$
	\end{algorithmic}
\end{algorithm}
\end{figure}

We perform the optimization over $\Sigma^n_s$, 
Algorithm~\ref{alg:ALSBT}~Step~\ref{alg:ALSBT:xi}, using the standard IHT algorithm for sparse vector recovery. 
Note that calculating the linear measurement map for $\xi$ given a fixed $\rho$ simply involves evaluating all calibration measurement blocks individually, 
i.e., calculating $\mathcal A_i(\rho)$ for all $i \in [N]$. 
Analogously, the low-rank optimization over $\CC^{d \times d}_r$, 
Algorithm~\ref{alg:ALSBT}~Step~\ref{alg:ALSBT:rho}, can be performed with iterative hard-thresholding on the manifold of low-rank matrices. A detailed description of a suitable algorithmic implementation is given by Algorithm~\ref{alg:SDTalgCompact} in the special case of a single matrix block, i.e., $N, s = 1$. 
Computing the corresponding linear map acting on $\rho$ for fixed $\xi$ amounts to summing up the individual 
measurement blocks weighted by their corresponding calibration coefficient.

The ALS optimization requires an initialization with a suitable $\rho^0$ in order to evaluate the first objective function for optimizing $\xi$. 
One method that we found viable is to  randomly draw a rank-$r$ state using Haar-random eigenvectors. 
Note that, in general, constrained ALS optimization can be highly sensitive to the chosen initialization. 
For this reason, depending on the measurement map and calibration model, alternative initialization strategies might become necessary. 
As break-off criteria we can again use a bound on the objective function as in \eqref{eq:stoppingcriterion} and an allowed maximal number of iterations.

\section{Recovery guarantees}\label{sec:guarantees}

We now prove that for certain simple measurement ensembles, the SDT algorithm converges to the optimal solution before we numerically demonstrate its performance in the following section. 
More precisely, following the outline of model-based compressed sensing \cite{BaraniukCevherDuarteHegde:2010, BluemsathDavies:2009}, the SDT algorithm can be accompanied by recovery guarantees based on a \emph{restricted isometry property} (RIP) of the measurement 
ensemble that is custom-tailored to the structure at hand. 
Intuitively, it seems clear that a measurement map should at least in principle allow for solving the associated linear inverse problem uniquely if it acts as an isometry on signals from the constraint set. 
So-called RIP constants formalize this intuition:
\begin{definition}[$\hat\Omega^{n,d}_{s,r}$-RIP]\label{def:RIP}
	Given a linear map $\mc A: \CC^{nd^2} \to \CC^{m}$, we denote by $\delta_{s,r}$ the smallest $\delta \geq 0$ such that 
	\begin{equation*}
	  (1-\delta)\|x\|_F^2 \leq \|\mc A(x)\|_{\ell_2}^2 \leq (1+\delta) \|x\|_F^2 
	\end{equation*}
	for all $x \in \hat\Omega^{n,d}_{s,r}$.
\end{definition}

The constant $\delta_{s,r}$ measures how much the action of $\mc A$ when restricted to elements of $\hat\Omega^{n,d}_{s,r}$ deviates from that of an isometry. 
Correspondingly, if $\delta_{s,r}$ is sufficiently small we expect this to be sufficient to ensure that the restricted action of $\mc A$ becomes invertible. 
In fact, if a measurement map has a sufficiently small RIP constant one can prove the convergence of projective gradient descent algorithms to the correct solution of the structured linear inverse problem. 
For the sake of simplicity, we analyze the SDT algorithm omitting the tangent space projection 
and also assuming a constant step widths $\mu^l = 1$. 
In numerically simulations we observe that making use of the tangent space projection and a more sophisticated heuristic for the step width yields faster convergence and better recovery performance. 
But the RIP assumption is in fact strong enough to already for this simpler algorithmic variant ensure that the following theorem holds:
\begin{theorem}[Recovery guarantee]\label{thm:recGarant}
	Let $\mathcal A: \CC^{nd\times d} \to \CC^m$ be a linear map and 
	suppose that the following RIP condition for $\mathcal A$ holds
	\begin{equation}\label{eq:recGarant:RIP}
		\delta_{3s,3r} < \frac{1}{2}.
	\end{equation}
	Then, for $X \in \hat\Omega^{n,d}_{s,r}$,  
	 the sequence $(X^l)$ defined by the SDT algorithm (Algorithm~\ref{alg:SDTalgCompact}) with $\mu^l=1$ and $P_{\mathcal T_{X^l}}=\Id$ with $y = \mc A(X)$ satisfies, for any $l\geq 0$, 
	\begin{equation*}
		\fnorm{X^l - X
		} \leq \gamma^l \fnorm{X^0 - X
		}, 
	\end{equation*}
	where 
	$
		\gamma =  2 \delta_{3s,3r} < 1
	$
\end{theorem}
{The theorem's proof is presented in Appendix~\ref{sec:convergence:proof}.}
We establish that the SDT algorithm converges to the correct solution of the 
sparse de-mixing problem at a rate that is upper bounded by the RIP 
constant $\delta_{3s,3r}$ of the measurement map.
The right-hand side of the RIP condition \eqref{eq:recGarant:RIP} is not expected to be optimal. 
Typically, one can at least improve the bound to $\frac1{\sqrt{3}}$ with a slightly more 
complicated argument \cite{FoucartRauhut:2013}. Since we are interested in the parametric 
scaling here, we choose to present a simpler argument at the cost of slightly worse constants. 
Furthermore, the statement of Theorem~\ref{thm:recGarant} does not account for statistical 
noise or potential mild violation of the signal constraints. 
{When deployed in practice, one of course expects that the calibration parameter and the quantum state will only be approximately sparse and of approximately low-rank, respectively. For example, 
even a small amount of depolarizing noise causes a pure quantum state to be of full rank. 
Such a state will still be well-approximated by a rank-one matrix, however, and one would expect the recovery to be robust to such deviations. 
Theorem~\ref{thm:recGarant} can be generalized to a noise- and model-robust guarantee.
We provide a detailed discussion at the end of Appendix~\ref{sec:convergence:proof}.} 
For the current analysis focusing on the scaling behaviour, we are content with the significantly simpler version. 

The pressing next question is, of course, which measurement ensembles actually exhibit the required RIP. 
Interestingly, it is notoriously hard to give deterministic constructions of measurement maps that are sample optimal and feature the RIP. 
In fact, already for the RIP for $s$-sparse vectors there are no sample optimal deterministic measurement maps known to date \cite{FoucartRauhut:2013}.
To further complicate the state of affairs, it is also known to be \NP{}-hard to check whether a given measurement map exhibits the $s$-sparse RIP with RIP constant small than a given $\delta$~\cite{bandeira_certifying_2013}.

For this reason, the field of compressed sensing uses probabilistic constructions to arrive at provably sampling optimal measurement maps.
Using a random ensemble of measurement maps of sampling optimal dimension one establishes 
that with high probability a randomly drawn instance will exhibit the RIP property. 
In other words, one proves that the originally hard linear inverse problem typically becomes easy for a certain measurement ensemble. 
Arguably, the simplest measurement ensemble consists of observables given by i.i.d.\ chosen random Gaussian matrices.
In our setting a fully Gaussian measurement map can be constructed from a set of $\{A_i \in \RR^{nd\times d}\}_{i=1}^m$ of $m$ Gaussian matrices with entries draws i.i.d.\ from the normal distribution $\mathcal N (0,1)$ and defining 
$
	y^{(l)} = \Tr(A_i X)
$.

As a toy model for quantum tomography it is more natural to consider observables drawn from a random ensemble of Hermitian matrices such as the Gaussian unitary ensemble ($\operatorname{GUE}$). Operationally, we define the $\operatorname{GUE}$ by drawing a matrix $X$ with complex Gaussian entries, $X_{k,l} \sim \mathcal N(0,1) + \i \mathcal N(0,1)$, and subsequently projecting $X$ onto Hermitian matrices using $P_\scrsym{}: X \mapsto \frac12 (X + X\ad)$.
For measurement maps from $\operatorname{GUE}$ we prove the following statement:
\begin{theorem}[$\hat\Omega^{n,d}_{s,r}$-RIP for random Hermitian matrices.]\label{thm:OmegaRIPGauss}
Let $\{ {A}^{(k)}_{i}\}_{i=1, k=1}^{n,m}$ be a set of Hermitian matrices drawn i.i.d.\ from the $\operatorname{GUE}$. 
Let $\mathcal{A}$ be the measurement operator defined by $\{ {A}^{(k)}_{i}\}_{i=1, k=1}^{n,m}$ via Eqs.~\eqref{eq:observed_data} and \eqref{eq:measurement_observable}. 
Then $\frac{1}{\sqrt{m}}\mathcal{A}$ satisfies the $\hat\Omega^{n,d}_{s,r}$-RIP with parameter $\delta_{s,r}$ with probability at least $1-\tau$ provided that 
\begin{equation}
	\label{eq:rip measurements}
	m \geq \frac{C}{\delta_{s,r}^2}\left[ s \ln\frac{\e n}{s} + (2d+1)rs \ln\frac{c}{\delta} + \ln \frac2\tau \right]
\end{equation}
for sufficiently large numerical constants $C, c>0$. 
\end{theorem}
The proof of the theorem is provided in Appendix~\ref{app:RIPguarantees}.
Based on the result for random Hermitian measurement maps we now discuss the asymptotic scaling of the measurement complexity of our approach to the blind tomography problem and the sparse de-mixing problem.
First, the derived measurement complexity \eqref{eq:rip measurements} is in accordance with the degrees of freedom of signal $X \in \hat{\Omega}^{n,d}_{s,r}$. 
The second term of $O(drs)$ corresponds to the number of degrees of freedom specifying the $s$ rank-$r$ matrices of dimension $d$. 
The first term of $O(s\ln n)$ is the minimal 
sampling complexity in $s$ for learning the $s$ non-trivial entries and their support \cite{FoucartRauhut:2013}.
Second, in analogy, we expect the optimal number of measurements for the blind tomography problem, i.e.,\ reconstructing signals in $\Omega_{s,r}^{n,d}$ instead of $\hat\Omega_{s,r}^{n,d}$, to scale as $O(s\ln n + dr)$. 
Hence, having a provably efficient algorithms capable of solving the blind tomography as well as the sparse de-mixing problem comes at the cost of an increase in the sampling complexity by an additional factor of $s$ in the second term of the sampling complexity. 
Most importantly, invoking the sparsity assumption on the calibration vector $\xi$ allows us to get away without a linear increase $n$ of the number of calibration parameters. 
Thus, the overhead in measurement complexity of our approach to the blind tomography problem is relatively mild. 

In fact, the measurement complexity derived for Gaussian measurements can often be used as a guideline for the sampling complexity of other measurement ensembles that are also sufficiently unstructured. 
However, the proof techniques for model-based compressed sensing that exploit the combination of different structures are not easily translatable to other measurement ensembles. 
An exception are measurement ensembles that feature a structure that is sufficiently aligned with the signal structure such as the one exploited in Ref.~\cite{RothEtAl:2016} for hierarchically sparse signals. 
We leave the study of more involved measurement ensembles to future work. 

\section{Numerical results}\label{sec:numerics}
The analytical results of the previous section provide worst-case bounds on the asymptotic
scaling for a class of 
idealized, unstructured measurements.
In order to benchmark and assess the non-asymptotic performance of compressed sensing algorithms in practice, however, numerical simulations are indispensable. 
In a first step we therefore perform numerical simulations for the idealized measurement model as given by random $\operatorname{GUE}$ matrices, 
comparing the performance of our algorithm to related established algorithms that do not entirely exploit the structure of the problem.
In a second step, we compare the SDT algorithm~\ref{alg:SDTalgCompact} with standard CS tomography in a blind tomography setting involving measurements of Pauli correlators, cnf.~\eqref{eq:targetSO}. 
To do so we randomly draw subsets of the possible Pauli measurements as possible calibrations $\mc A_i$ of the measurement apparatus. 
Finally, we demonstrate the feasibility of blind tomography under structure assumptions in the realistic measurement and calibration setting involving single-qubit coherent errors described in Section~\ref{sec:application}. 
To this end, we employ the Algorithm~\ref{alg:ALSBT} that performs alternating constrained optimization.
The algorithms and the scripts producing the plots have been implemented in Python and will be made available under Ref.~\cite{numerics_repo}.

\begin{figure}[tb]
	\includegraphics[width=.5\textwidth]{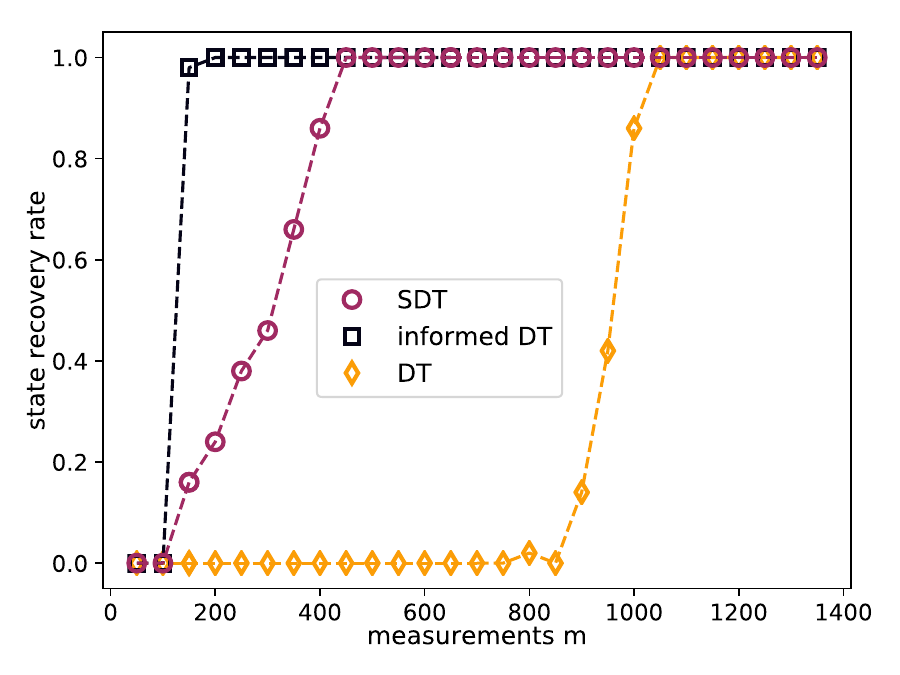}
	\caption{
	\label{fig:numericsGauss} 
	The recovery rate for the SDT, DT and informed DT algorithm for different number of observables $m$ for GUE measurements. Each point is averaged over $50$ random measurement and signal instances with $r=1$, $d=16$, $n=10$ and $s=3$.
	A signal is considered successfully recovered if its Frobenius norm deviation from the original signal is smaller than $10^{-3}$. 
	One observes nearly coinciding recovery performances for the informed DT and the SDT algorithm. In
	 comparison, the DT algorithm requires significantly more observables for recovery. 
	}
\end{figure}

\begin{figure*}[tb]
    \centering
\includegraphics[width=1\textwidth]{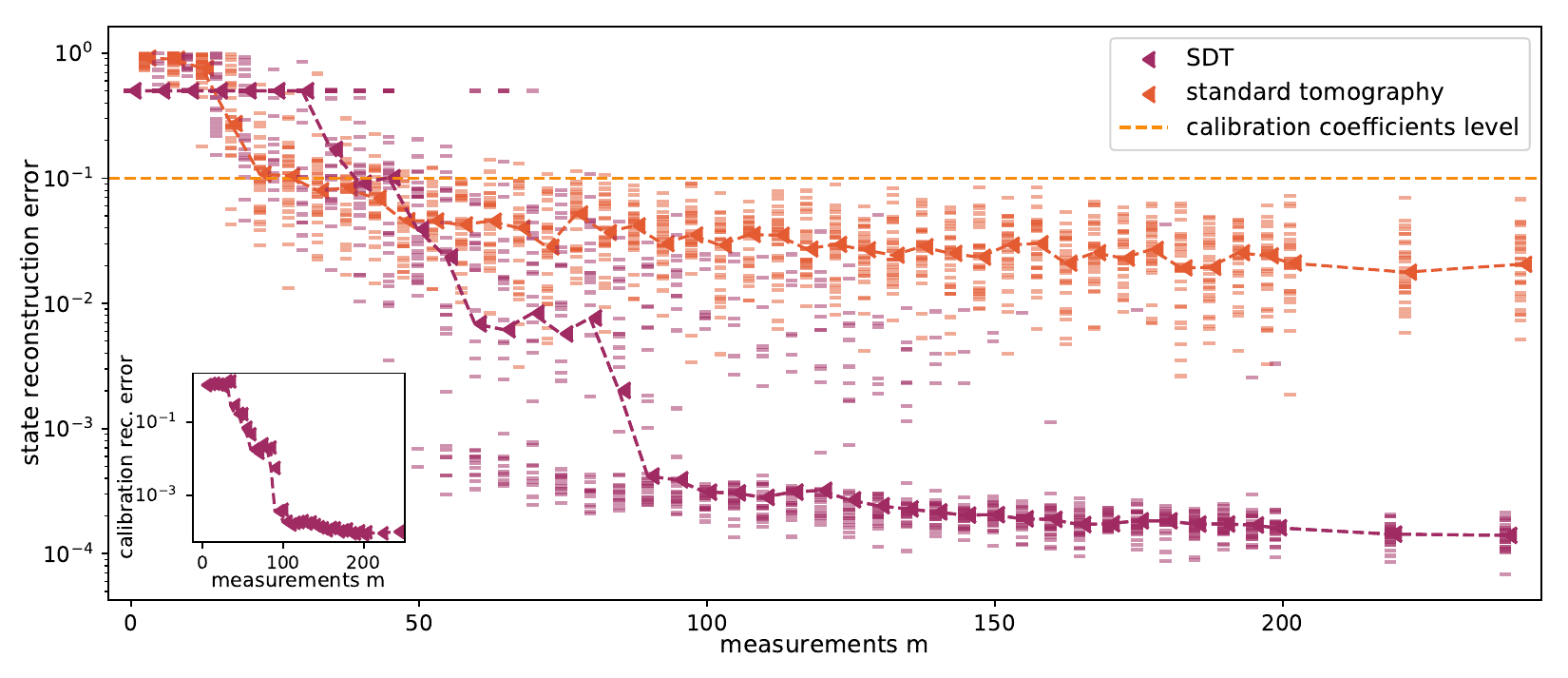}
	\vspace{-.7cm}
	\caption{ 
	\label{fig:numericsrandompauli}
	The trace norm reconstruction error for the SDT compared to the standard tomography algorithm for different number of observables $m$ for sub-sampled random Pauli measurements. Each point is averaged over $30$ random measurement and signal instances with $r=1$, $d=8$, $n=10$ and $s=3$. The inline figure shows the mean $\ell_2$-norm reconstruction error of the calibration coefficients for the SDT algorithm.
	}
\end{figure*}

\subsection{\texorpdfstring{$\operatorname{GUE}$}{GUE} measurements}

The SDT algorithm goes beyond existing IHT algorithms for the de-mixing problem 
of low-rank matrices in that it additionally allows one to exploit a sparse mixture. 
We demonstrate that this yields a drastic and practically important improvement 
in the number of measurement required for the reconstruction. 

To this end, we draw signal instances $X = \xi \otimes \rho$ at random from $\Omega^{n,d}_{s,r}$. 
We use four qubit pure states $\rho = \ketbra\psi\psi$ with $r=1$ and $d=16$, where $\ket\psi$ is drawn uniformly (Haar) random from the complex $\ell_2$-norm sphere. 
The calibration vector $\xi \in \RR^n$ with $n=10$ has a support of size $s=3$ drawn uniformly from the set of all $\binom{n}{s}$ possible supports. The non-vanishing entries of $\xi$ are normal distributed with unit variance.
The measurements are drawn at random from the $\operatorname{GUE}$ ensemble as defined above with a varying number of observables $m$. 

The closest competitor to the SDT algorithm is the related algorithm of Ref.~\cite{StrohmerWei:2017}. 
The algorithm of Ref.~\cite{StrohmerWei:2017} coincides with the special case of the SDT 
algorithm where we use the projection on to $\hat\Omega^{n,d}_{n,r}$ with $s=n$ ignoring 
the sparsity in the block structure. We will refer to this algorithm as the \emph{DT algorithm}.
We can also give the DT algorithm the `unfair' advantage of restricting the problem to the correct block support of the signal from the beginning. 
We will refer to this variant as the \emph{informed DT algorithm}. 

Figure~\ref{fig:numericsGauss} shows the recovery rate for the SDT algorithm, the DT algorithm and its informed variant for different $m$. 
Each point is average over $50$ random signal and measurement instances. 
We consider a signal as successfully recovered if the distance of the algorithm's output to the original signal is smaller than $10^{-3}$ in Frobenius norm. 
The algorithm terminated if either the stopping criterion~\eqref{eq:stoppingcriterion} with $\gamma_\text{break} = 10^{-5}$ is met or after a maximal number of 600 iteration.
We observe that if one of the algorithm successfully 
recovers a signal it typically meets the stopping criterion after less than 100 iterations.

The curves for all three algorithm in Figure~\ref{fig:numericsGauss} 
display a sharp phase transition
from a regime where the number of measurement is too 
small to recover any signal to a regime of reliable recovery. 
While the phase transition for the SDT algorithm appears in a similar regime to the informed DT algorithm, the DT algorithm requires considerably more samples in order to recover the signal instances. 

We conclude that the sparsity of the calibration parameters can be exploited by the SDT algorithm to considerably reduce the required number of measurements. 
Even more so, this does not require many more sampling points as compared to an algorithm which is given \emph{a priori} knowledge which errors were present, that is, the block support of the signal. 
This shows that the SDT algorithm solves the de-mixing and blind tomography task in a highly efficient way and scalable. 
Finally, the number of possible erroneous measurements $\mc A_i$ can be scaled up at a very low cost in terms of required measurement settings.

\begin{figure*}[tb]
    \centering
\includegraphics[width=1\textwidth]{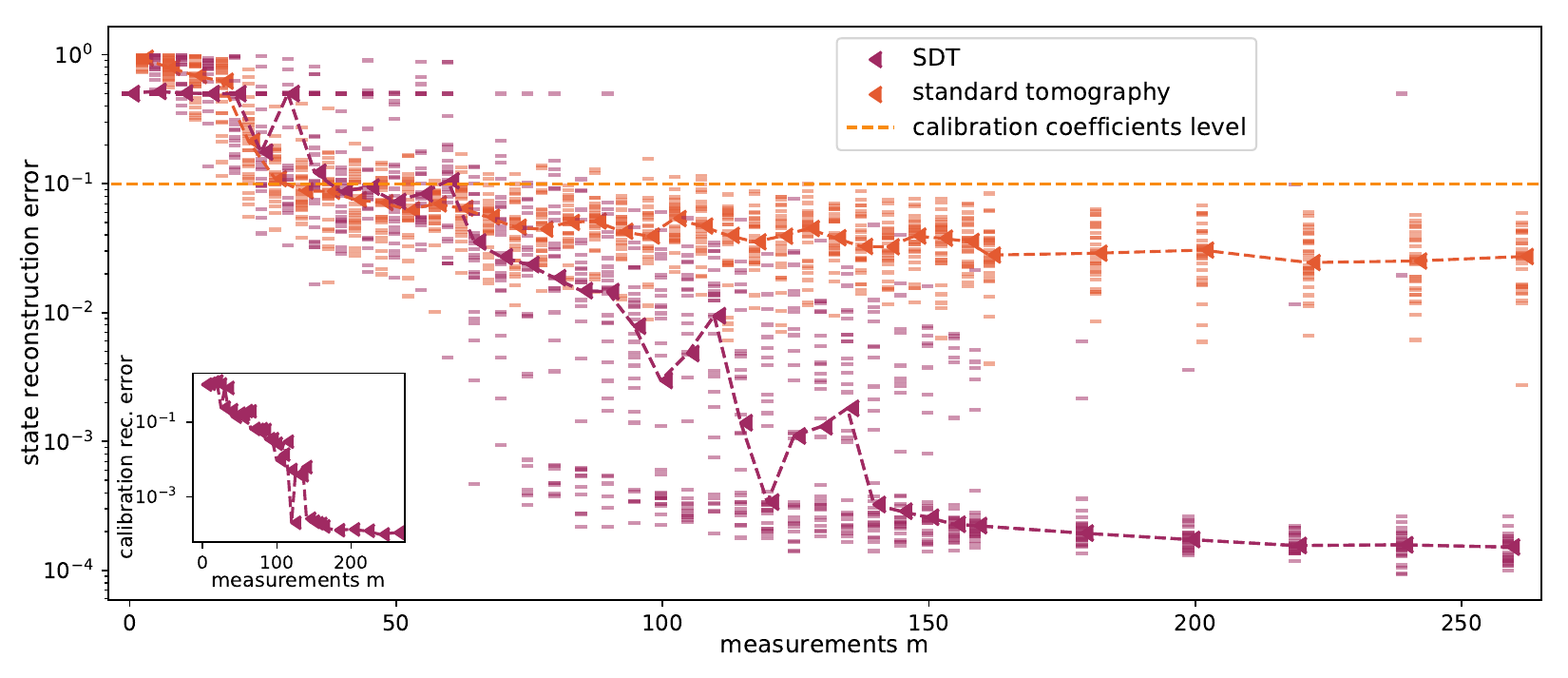}
	\vspace{-.7cm}
	\caption{ 
	\label{fig:numericsrandompauli2}
	The trace norm reconstruction error for the SDT compared to the standard tomography algorithm for different number of observables $m$ for sub-sampled random Pauli measurements. Each point is averaged over $30$ random measurement and signal instances with $r=1$, $d=8$, $n=10$ and $s=4$. The inline figure shows the mean $\ell_2$-norm reconstruction error of the calibration coefficients for the SDT algorithm.
	}
\end{figure*}

\subsection{Sub-sampled Pauli measurements}

For the application in characterizing quantum devices, it is key to compare the recovery performance of the SDT algorithm with standard low-rank quantum tomography algorithms. To this end note that the SDT algorithm restricted to $n, s=1$ is also a state-of-the-art algorithm for standard low-rank state tomography without the on-the-fly calibration. 
Thus, we will make use of this implementation of conventional low-rank state tomography in the following.

We draw signal instances as before but using three-qubit states, $s \in \{3,4\}$ and altering the model for the calibration parameter: 
We set the first entry of $\xi$ to $\xi_0 = 1$. The support of the remaining entries is drawn uniformly at random. 
The non-vanishing entries are then i.i.d.\ taken from the normal distribution rescaled by a factor of $1/10$. 
This mimics a setting where we have a dominant target measurement and a couple of small systematic deviation from a known set of candidates. The target measurements as well as the systematic deviations are uniformly sub-sampled Pauli observables. Thus, $\mathcal A_0$ till $\mathcal A_n$ have the form of \eqref{eq:targetSO} with differently i.i.d.\ selected Pauli observables uniformly selected from $\{\Id, X, Y, Z\}$.
We simulate statistical noise using $10^{8}$ samples per expectation value in order to realistically limit the resolution of the SDT algorithm.

We simultaneously perform recoveries with the SDT algorithm using the entire measurement matrix including the calibration measurement components and the SDT algorithm using only the target measurement $\mathcal A_0$ as in a conventional tomography setting.

The resulting trace distance of the state estimate, i.e.,\ the trace-normalized first block of $X$, from the original $\rho$ is shown for different number of measurements in Figure~\ref{fig:numericsrandompauli} and Figure~\ref{fig:numericsrandompauli2} for different sparsity $s = 3$ and $s= 4$, respectively. The curves indicate the median over the depicted 30 sample points per value of $m$.
The inline plot of both figures show the $\ell_2$-norm deviation of the reconstructed calibration parameters and the original $\xi$.

\begin{figure*}[tb]
	\includegraphics[width=1\textwidth]{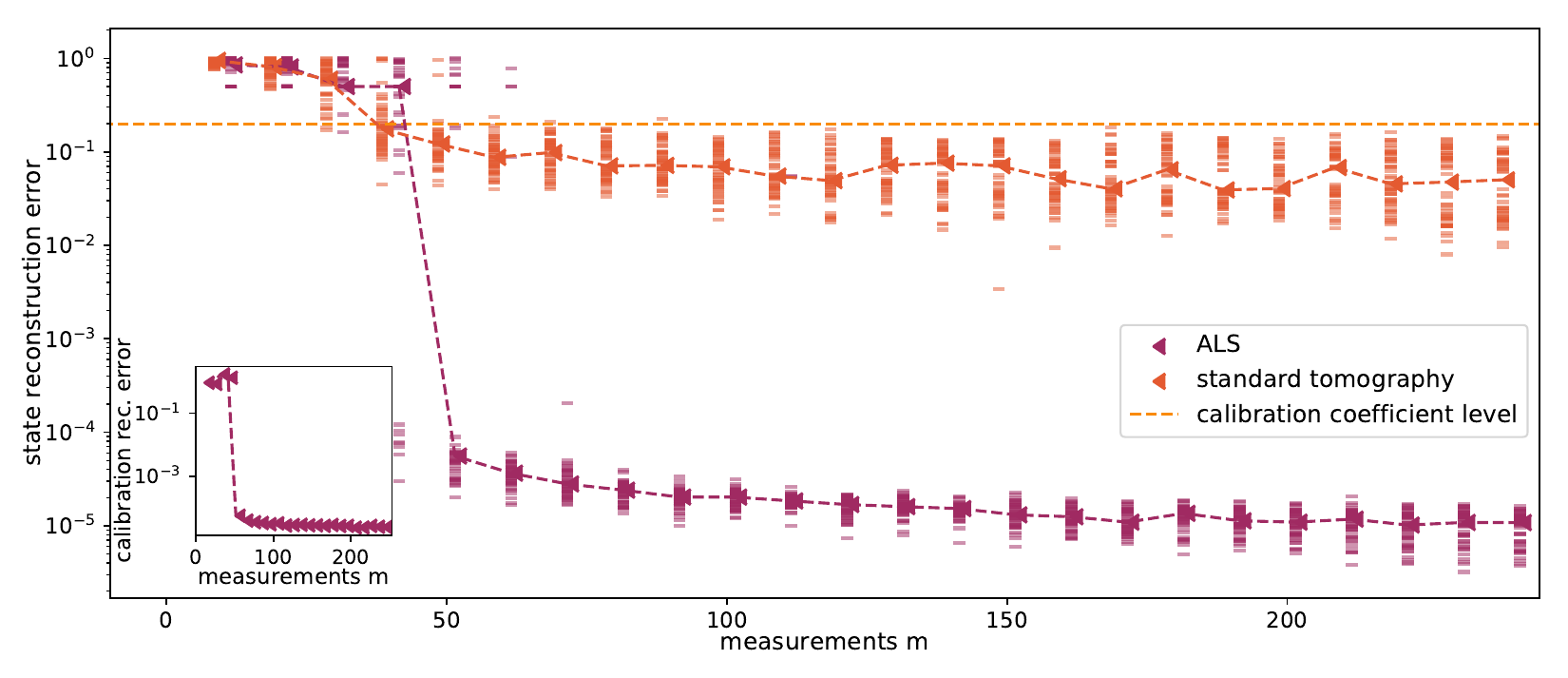}
	\vspace{-.7cm}
	\caption{\label{fig:ALSsparsity2}The trace norm reconstruction error for the ALS compared to the standard tomography algorithm for different number of observables $m$ for Pauli measurements with coherent single-qubit errors. Each point is averaged over $50$ random measurement and signal instances with $r=1$, $d=16$, $n=7$ and $s=2$. The inline figure shows the mean $\ell_2$-norm reconstruction error of the calibration coefficients for the {ALS} algorithm.}
\end{figure*}

One observes that the conventional low-rank tomography becomes more accurate with an increasing number of measurement but is asymptotically still bounded from below by the systematic error induced by the calibration on the order of $10^{-1}$. 
This agrees with the order of magnitude of variance of the calibration coefficients. 
In contrast, the SDT algorithm while performing slightly worse in a regime of insufficient measurements outperforms the conventional algorithm for a moderate number of samples and is ultimately only limited by the statistical noise.
However, in the parameter regime under investigation there are even for large number of samples $m > 150$ a small number (well below 10\%) of instance where SDT only reaches an accuracy comparable to standard tomography. 
In these instances we find that the support for the calibration measurement components was incorrectly identified. For $s=4$ we furthermore observe one pathological instance of SDT for $m=240$ that is worse in recovery than standard tomography is in this regime. 
For $s=4$ the phase transition of SDT appears for a slightly larger values of $m$ compared to $s=3$. 
The curves for the reconstruction error of the quantum state approximately coincide with the curves for the error in the calibration parameter. 
We conclude that for a sufficient number of measurement settings, the SDT algorithm almost always performs a significantly more accurate state reconstruction and simultaneously extracts the calibration parameters. 
The precision is ultimately only limited by the statistical error in the estimation of the expectation values. 
\begin{figure*}[tb]
	\includegraphics[width=1\textwidth]{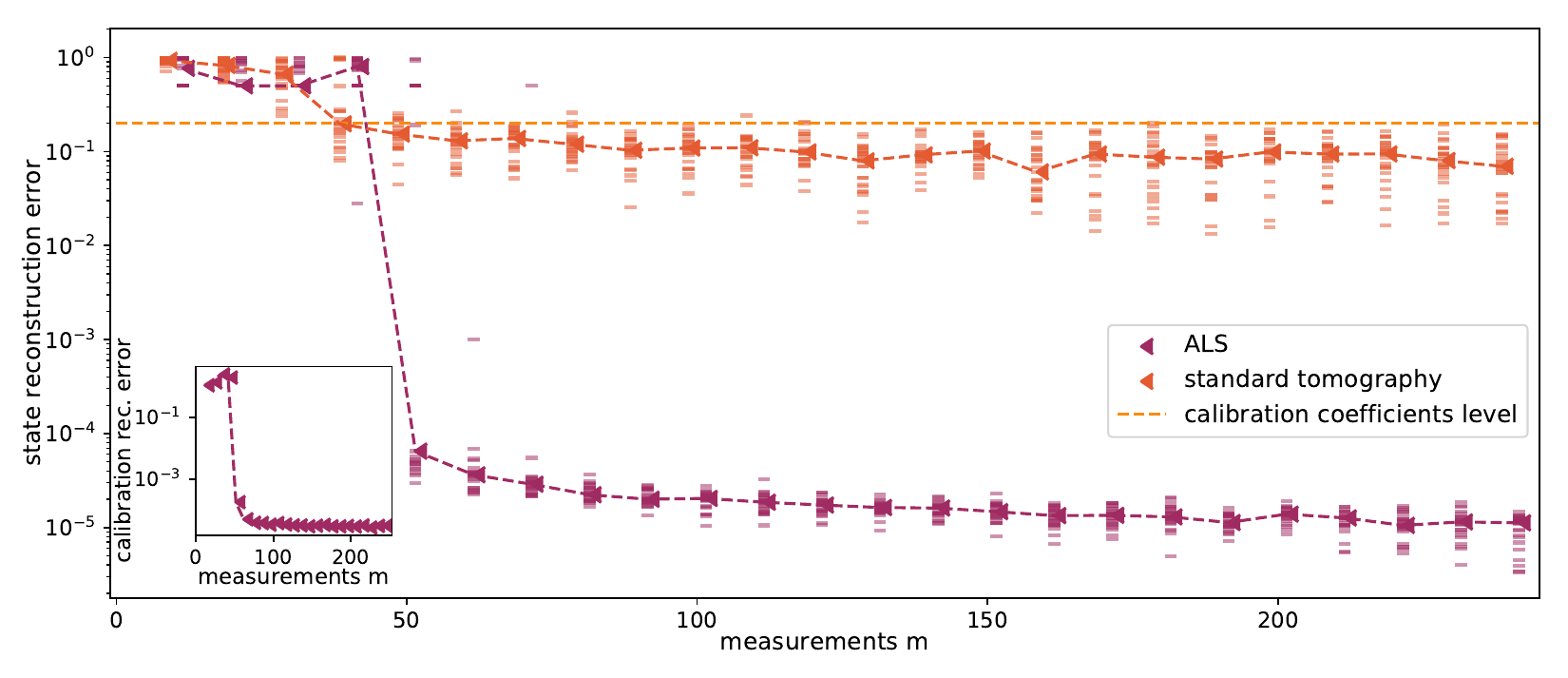}
	\vspace{-.7cm}
	\caption{\label{fig:ALSsparsity3}The trace norm reconstruction error for the ALS compared to the standard tomography algorithm for different number of observables $m$ for Pauli measurements with coherent single-qubit errors. Each point is averaged over $30$ random measurement and signal instances with $r=1$, $d=16$, $n=7$ and $s=3$. The inline figure shows the mean $\ell_2$-norm reconstruction error of the calibration coefficients for the {ALS} algorithm.}
\end{figure*}

\subsection{Pauli measurements with coherent single-qubit errors}
We now come back to the concrete realistic scenario described in Section~\ref{sec:application}. 
There we derived the calibration measurement model originating from coherent errors in the gates that implement the single-qubit measurements. 

For the numerical simulations, we draw a set of $m$ Pauli observables uniformly at random as the target measurement. 
Subsequently, we introduce six calibration blocks such that every observable in the set $\{X, Y, Z\}$ is swapped with another Pauli observable in $\{X, Y, Z\}$ in a specific block. 
We generate data $y$ for given states and calibration parameters using the linear calibration measurement model without noise as induced by finite statistics.

We find that in the parameter regimes that are easily amenable to numerical studies on desktop hardware the SDT~algorithm is not capable of successfully reconstructing the states when the calibration parameters for the corrections are considerably smaller than the leading order measurement. 
To thoroughly understand this limitation, in the following, we briefly report the performance of the SDT algorithm on different sub-tasks related to the recovery problem.

First, we choose $d= 16$ and $n=s=1$ such that only a single block, either the ideal measurement or one of the correction blocks, is used to generate the signal from a random pure state ($r=1$).
We observe that the SDT algorithm is able to recover the signals in this standard tomography problem. 
This indicates that also the calibration blocks individually allow for tomographic reconstruction of low-rank states.
Second, the SDT algorithm can discriminate between different mixtures of the six correction blocks. 
To demonstrate this, we ignore the ideal measurement and employ only the correction blocks to generate the signal. 
We set the active calibration coefficients to one. Thus, $n=6$, $s\leq n$ and $\xi_i = 1$ for $i$ active.
We observe that given a sufficient number of measurement settings the SDT algorithm correctly reconstructs pure states in this measurement setting.
The same findings hold true if the target measurement is again considered as long as the active calibration coefficients are set to $1$. 
We observe successful reconstructions of unit rank states for $n=7$ and $s\in{1,2,3}$.

A more natural setting however would typically have calibration coefficients that are considerably smaller than the ideal measurement. 
This justifies the linear expansion for the measurement model in the first place. 
If we choose, e.g., $\xi_i = 1/10$ for the indices $i$ of active blocks, we were unable to identify a parameter regime on desktop hardware where the SDT algorithm can successfully recover the majority of instances of pure states. 
We observe that if the SDT algorithm settles on an objective variable with an incorrect block support in the first few iterations, 
it is not able to subsequently run into objective variables with a different block support in most instances. 
Despite the negative result for the SDT algorithm in the most realistic setting, the general mindset to exploit structure (low-rankness) to allow quantum state tomography in a blind fashion is fruitful using a slightly different algorithmic strategy. 

To this end, we use the constrained \emph{alternating least square} (ALS) algorithm described in Section~\ref{sec:algorithm:als}. 
We set the first calibration coefficient corresponding to the ideal measurement to one. The support of the remaining active calibration coefficients is drawn uniformly at random and their value are i.i.d.\ drawn from a shifted normal distribution with standard deviation $0.05$ and mean value $0.2$. 
We use Haar random pure states, $r=1$ of a four-qubit system, $d=16$, as the target states. 

The algorithm is initialized with a Haar-randomly drawn pure state. We allow for a maximal number of $1000$ iterations of the algorithm or terminate if the criterion \eqref{eq:stoppingcriterion} with $\gamma_{\text{break}} = 10^{-5}$ is met. 
Furthermore, if the stopping criterion is not met after
$50$ iterations, we re-initialize the algorithm with a new random pure state. 
We allowed for a maximal number of $10$ or $20$ re-initializations for $s=2$ and $s=3$, respectively. 
We observe that in case of successful recovery typically at most $3$ re-initializations are required with most instances already correctly converging from the initial state.

As in the previous section, we compare the recovery performance of the ALS with the standard low-rank tomography algorithm.
The trace-norm error and calibration error for different numbers of measurement settings for $s=2$ and $s=3$ are displayed in Figure~\ref{fig:ALSsparsity2} and \ref{fig:ALSsparsity3}, respectively.
We observe that, as expected, the reconstruction error of standard low-rank tomography is again lower-bounded by a scale set by the magnitude of the calibration parameters. In contrast, with an only slightly larger number of measurement settings, the constrained ALS algorithm is capable of recovering the states and the calibration parameter with an accuracy that is improved by orders of magnitude and in the noiseless scenario only limited by the algorithms stopping criterion.  
Compared to recovery performance of the SDT algorithm we observe an even sharper phase transition to the regime of recovery.

{
Finally, in practice due to further imperfections, e.g.\ incoherent noise, the underlying state and calibration vector will actually be only approximately of low-rank and approximately sparse, respectively.
In order to probe the robustness of the ALS algorithm against such model mismatch, we generate measurement data as before for $s=2$, choosing $m=130$ and either add depolarizing noise to the state of different strength or replace the vanishing coefficient in the calibration vector by the absolute value of random Gaussian variables of different standard deviation.
Figure \ref{fig:model_mismatch} displays the reconstruction error of the state and calibration coefficients against the amount of model mismatch, as determined by the depolarizing strength and Gaussian widths. 
We observe a linear dependency of the reconstruction errors with the model mismatch and conclude that the method is robust enough to tolerate small deviations from the structure assumptions.
}

\begin{figure}[tb]
	\includegraphics[width=.5\textwidth]{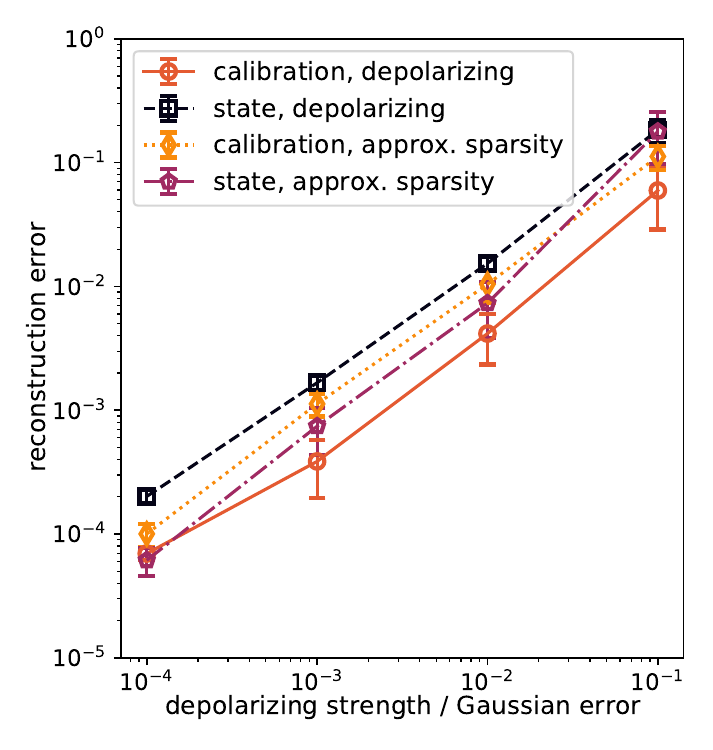}
	\caption{
	\label{fig:model_mismatch}
	{
	The reconstruction error of the state and calibration coefficients for the 
	ALS algorithm when adding model mismatch in terms of depolarizing noise on the state or Gaussian noise to the zero entries of the calibration coefficients. 
	Each point is averaged over $20$ random measurement and signal instances with $r=1$, $d=16$, $n=10$, $s=2$ before adding the model mismatch. 
	We use $m=130$ Pauli measurements with coherent single-qubit errors. The error bars indicate a half of the standard deviation. 
	The reconstruction error increases proportional to the model mismatch indicating robustness of the method.}
	}
\end{figure}

\section{Summary and outlook}\label{sec:outlook}

In this work, we have shown that the natural assumption of low-rankness allows one to perform self-calibrating quantum state tomography. 
Relaxing the blind tomography problem to a sparse de-mixing problem has allowed us to develop an efficient classical post-processing algorithm, the SDT algorithm, that is theoretically guaranteed to recover both  
the quantum state and the device calibration under a restricted isometry condition of the measurement model. 
We have demonstrated the necessity of relaxing the blind tomography problem within the framework of hard-thresholding algorithms by establishing the \NP{}-hardness of the projection onto the set consisting of the outer products of vectors and fixed-rank matrices. 
Introducing a sparsity assumption on the calibration coefficients ensures that the reconstruction scheme can already be applied for fairly small system dimension. 
We have explicitly proven that a Gaussian random measurement model meets the required restricted isometry condition with a close-to-optimal measurement complexity in $O(s \ln n + d rs)$. 
Furthermore, we have numerically demonstrated an improved performance of the SDT algorithm for random instances of measurement models compared to previously proposed non-sparse de-mixing algorithms and standard low-rank state tomography.

While these generic measurement and calibration models allows us to derive analytical guarantees, it is fair to argue that these models might at best capture some aspects of actual experimental implementations.  
A potential starting point for extending recovery guarantees to more realistic settings 
is the generalization of our results to random Pauli measurements as considered in Sec.~\ref{sec:numerics} \cite{Liu:2011} together with the coherence measures and structured measurement guarantees developed in the context of hierarchically spares signals \cite{SprechmannEtAl:2011,RothEtAl:2016,RothEtAl:2018:Kronecker,HierarchyIEEE}.

To complement our conceptually and rigorously minded work with a more pragmatic approach, we have  additionally developed and implemented a structure-exploiting blind tomography algorithm based on alternating optimization. 
We have numerically demonstrated that the alternating algorithm is able to perform self-calibrating low-rank tomography in a realistic measurement and calibration model that is well-motivated by gate implementations in ion traps.
These numerical simulations indicate that the approach to the blind tomography problem developed here might be well-suited 
to improve tomographic diagnostics in current experiments. 
Ultimately, the recovery performance of the proposed algorithms has to be evaluated on experimental data. 
It is the hope that this work contributes to establishing a new mindset in quantum system identification 
and specifically tomographic recovery in which no component used has to be precisely known, but still under
physically meaningful structural assumptions, a mindset
here referred to as being semi-device-dependent.

\appendix
\section{Hardness of projection}\label{app:hardnessproof}

As a starting point we state the SparsePCA problem. 
\begin{problem}[SparsePCA]
	\emph{Input:} Symmetric matrix $A \in \RR^{n\times n}$, sparsity $s$, positive real number $a > 0$. 
	\emph{Question:} Does there exist an $s$-sparse unit vector $v \in \RR^n$ with $v^TAv \geq a$?
\end{problem}

It has been folklore for quite some time that the sparse PCA problem is \NP{}-hard. A formal proof can be found in Ref.~\cite{Magdon-Ismail:2017}, where the CLIQUE problem is encoded into instances of SparsePCA.
From the hardness of SparsePCA it follows that there does not exist a polynomial time algorithm for the projection onto the set of symmetric, unit rank matrices with sparse eigenvectors{, unless $\Poly{}=\NP$.} 
Formally, we have:
\begin{proposition}[Hardness of projection onto the set of symmetric, unit rank matrices with sparse eigenvectors]
Given a matrix $A \in \RR^{d\times n}$ and $s, \sigma \in \NN$, there exist no polynomial time algorithm that calculates
\begin{equation*}
\begin{split}
	& \operatorname*{minimize}\quad \| A - vw^T\|_F ,\\& \text{\upshape subject to $v \in \Sigma^d_\sigma$, $w\in \Sigma^n_s$}\,,
\end{split}
\end{equation*}
{unless $\Poly{}=\NP$.} 
This still holds for $\sigma = d$.
\end{proposition}
\begin{proof}
It turns out to be sufficient to only consider the case where $\sigma = d$, i.e.,\ only one of the factors is required to be sparse. 
It is straightforward to see that solving the problem with both vectors being sparse allows one to solve the projection with only one sparse vector:
Define
\begin{equation*}
	A = \begin{pmatrix}
		\vec 0_{d-\sigma, n} \\
		A'
		\end{pmatrix}
\end{equation*}
with $\vec 0_{a,b}$ being an $a\times b$ matrix filled with zeros. 
It then holds that 
\begin{equation*}
	\min_{v\in \Sigma^d_{\sigma}, w\in \Sigma^n_s}\| A - vw^T\|_F = \min_{v'\in \CC^\sigma, w\in \Sigma^n_s} \|A' - v'w^T\|_F. 
\end{equation*}
We now embed the SparsePCA problem. 
To do so we first make the normalization of the vectors $v, w$ in the optimization problem explicit to it to a maximization problem over normalized vectors: 
\begin{equation}\label{eq:onesparsefactor}
\begin{split}
&\min_{v \in \RR^\sigma ,w \in \Sigma^n_s} \|A - vw^T\|^2_F \\
= &\min_{\lambda \in \RR, v \in \RR^\sigma \cap B_{\ell_2}^{\sigma} ,w \in \Sigma^n_s \cap B_{\ell_2}^n} \|A - \lambda vw^T\|^2_F
\end{split}
\end{equation}
with $B_{\ell_2}^n = \{v \in \RR^n \mid \|v\|_{\ell_2} \leq 1 \}$ the $\ell_2$-norm ball. 
Solving the optimization problem over $\lambda$ yields
\begin{align*}
	& \min_{\lambda \in \RR}  \| A - \lambda v w^T\|^2_F \\
	&= \min_{\lambda \in \RR} \left\{ \|A\|_F^2 + \lambda^2 \|v\|_{\ell_2}^2 \|w\|_{\ell_2}^2 - 2 \lambda\langle w, Av \rangle \right\}\\
	&= \|A\|_F^2 - \min_{v \in \RR^\sigma \cap B_{\ell_2}^{\sigma} ,w \in \Sigma^n_s \cap B_{\ell_2}^n} \langle w, A v\rangle^2.
\end{align*}
Since $A$ is fixed we conclude that the optimization problem \eqref{eq:onesparsefactor} is equivalent to 
\begin{equation*}
\begin{split}
	&\qquad\maximize\ |\langle w, Av\rangle|\ \\
	&\text{subject to $v \in \RR^\sigma \cap B_{\ell_2}^{\sigma}$, $w \in \Sigma^n_s \cap B_{\ell_2}^n$}.
\end{split}
\end{equation*}
Furthermore, using the Cauchy-Schwarz inequality we find that
\begin{equation*}
	\max_{v \in \RR^\sigma \cap B_{\ell_2}^{\sigma}, w \in \Sigma^n_s \cap B_{\ell_2}^n} |\langle v, Aw\rangle| = \max_{w \in \Sigma^n_s \cap B_{\ell_2}^n} \|A w\|_{\ell_2}.
\end{equation*}
Now consider an instance of the SparsePCA 
problem with a symmetric input matrix 
$B \in \RR^{n\times n}$, sparsity $s$ and $a>0$. 
W.l.o.g.\ we can assume that $B$ is a positive matrix since solving the SparsePCA problem for the $B - \min{}\{0, \lambda_{\min{}}(B)\}\Id$ 
shifted by the smallest eigenvalue $\lambda_{\min{}}(B)$ of $B$ and $a$ shifted correspondingly, allows one to solve the SparsePCA problem for $B$. For a positive matrix $B$ we find a factorization $B = A^T A$. Hence, deciding whether the maximum over all $w\in \Sigma_s^n$ of $w^T B w$ is larger than $a$ is solved by calculating the maximum of $\|Aw\|_{\ell_2}^2 = w^T B w$.
This completes the reduction. 
\end{proof}

We are now prepared to tackle our related problem: the projection onto $\Omega^{n,d}_{s,r}$. We have the following statement: 
\begin{theorem}[Hardness of constrained minimization]
	There exist no polynomial time algorithm that calculates for all $A \in \CC^{n\times n}$:
	\begin{equation}\label{eq:OmegaProjection}
		\operatorname*{minimize}\quad \| A - X \|_F \quad\text{\upshape subject to $X \in \Omega^{n,d}_{s,r}$},
	\end{equation}
	unless $\Poly{} = \NP{}$.
	This still holds for $s=n$. 
\end{theorem}

We note that our result for exactly computing  straightforwardly generalizes to the case of approximating the target function up to constant relative error using results on the approximatability of SparsePCA~\cite{chan_approximability_2016}. 

\begin{proof}
	Suppose there existed an efficient algorithm that determines the objective value of the projection \eqref{eq:OmegaProjection}. To encode the SparsePCA problem, we choose an instance of $A$ as follows: Let $A' \in \RR^{n\times d}$ be a matrix and let $A'_i$ denote the $i$th row of $A$. Let $e_i$ be the basis vectors $(e_i)_j = \delta_{i,j}$, with $\delta_{i,j}$ the Kronecker symbol. We choose $A = \sum_{i=1}^n e_{i} \otimes \diag(A'_i)$, where $\diag(A'_i)$ denotes the diagonal matrix with the $i$th row of $A'$ on its diagonal. Furthermore, we define $a' = ( A'_1, \ldots, A'_n) \in \RR^{nd}$ to be the vector arising by concatenating all rows of $A$. 
	By definition an $X\in \Omega^{n,d}_{s,r}$ can be decomposed as $X = \xi \otimes \rho$ with $\xi \in \Sigma^n_s$ and $\rho \in H^d_r$. Let $\rho = U \diag(\lambda) U^\dagger$ the eigenvalue decomposition of $\rho$ with a suitable unitary $U \in U(n)$ and $\lambda$ the vector of its eigenvalues. Then, we can rewrite

	\begin{align*}
		\| A &- \xi \otimes \rho \|_{2}^{2} 
			=  \sum_{i=1}^{n} \| \diag(A'_{i}) - \xi_{i} \rho ) \|_{2}^{2}  \\
			&= \| \diag(a') - (\Id_n \otimes U) \diag(\xi \otimes \lambda) (\Id_n \otimes U^\dagger) \|_2^2 \\
			&= \| \diag(a') (\Id_n \otimes U) - (\Id_n \otimes U) \diag(\xi \otimes \lambda)\|_2^2 \\
			&= \sum_{i,j = 1}^{nd} |A'_i - (\xi\otimes \lambda)_j|^2 |(\Id_n \otimes U)_{i,j}|^2, 
	\end{align*}
	where we have used the unitary invariance of the $\ell_2$-norm in the third step. 
	We can introduce the doubly stochastic matrix $W$ with entries $W_{k,l} = |U_{k,l}|^2$ and relax the optimization to
	\begin{align}\label{eq:doublyStochasticIneq}
		&\min_{\xi \in \Sigma^n_s, \rho \in H^d_r} \|A - \xi \otimes \rho \|^2_2 \\
		&\quad \leq  \min_{W \in \mathrm{DS}^n,\, \xi \in \Sigma_n^s,\, \lambda \in \Sigma_r^d} \notag\\
		&\qquad\qquad\qquad\quad\sum_{i,j = 1}^{nd} |A'_i - (\xi\otimes \lambda)_j|^2 (\Id_n \otimes W)_{i,j},\notag
	\end{align}
	where $W$ is optimized over all doubly stochastic matrices $\mathrm{DS}^d \subset \CC^{d\times d}$. For $\sigma \in S_d$, a permutation of the symbols in $[d]$, we denote the corresponding permutation matrix by $\Pi_\sigma: \CC^d \to \CC^d$, $\xi \mapsto \Pi_\sigma \xi$ with $(\Pi_\sigma \xi)_i = \xi_{\sigma(i)}$. By Birkhoff's theorem, see e.g., 
	Ref.~\cite[Theorem II.2.3]{Bhatia:1997}, the set of extremal points of the convex set of  doubly stochastic matrices $\mathrm{DS}^d$ are the permutation matrices $\Pi_{S_d} = \{ \Pi_\sigma \mid \sigma \in S_d\}$. 

	Since the optimum is, hence, attained for a permutation matrix $W = \Pi_\sigma$ and $U_{i,j} = (\Pi_\sigma)^{1/2}_{i,j} = (\Pi_\sigma)_{i,j}$ is a unitary matrix, the inequality \eqref{eq:doublyStochasticIneq} is saturated. 
	Therefore, we conclude that  
	\begin{align*}
		&\min_{\xi \in \Sigma^n_s, \rho \in H^y_r} \| A' - \xi\otimes \rho \|^2_2 \\
		&\quad= \min_{\xi \in \Sigma^n_s, \lambda \in \Sigma^d_r, \sigma \in S_d} \|a' - \xi \otimes \Pi_\sigma \lambda \|_2^2 \\
		&\quad= \min_{\xi \in \Sigma^n_s, \lambda \in \Sigma^d_r } \|a' - \xi \otimes \lambda\|_2^2 \\
		&\quad= \min_{\xi \in \Sigma^n_s, \lambda \in \Sigma^d_r } \|A' - \xi \lambda^T\|_2^2.
	\end{align*}
	Thus, an algorithm calculating the projection onto $\Omega^{n,d}_{s,r}$ for the matrix $A$ 
	chosen here solves the SparsePCA problem for $A'$. We conclude that there exists no polynomial time algorithm for the problem. 
\end{proof}

\section{Convergence proof}
\label{sec:convergence:proof}
In this section we provide the proof of Theorem~\ref{thm:recGarant}. We first introduce a bit more notation. Consider $X \in \Omega^{n,d}_{s,r}$. 
By definition, it can be written as $X = \sum_{i=1}^n \xi_i e_i \otimes x_i$ with $\xi \in \Sigma^n_s$ and $x_i \in \mathcal D_r^d$ for all $i$. 
Let $Q_i$ be the projector onto the range of $x_i$. 
Furthermore, we set $Q_i = 0$ for all $i$ not in the support of $\xi$. 
Slightly overloading our notation, we define the projection of every `block' onto the range of the corresponding `block' of $X$ as
$\mc P_{\hat{\Omega}(X)}(Y) \coloneqq P_{\hat\Omega(X)} Y P_{\hat\Omega(X)}$ with
\begin{equation*}
	P_{\hat\Omega(X)} \coloneqq \diag(Q_1, \ldots, Q_n).
\end{equation*}
Note that the projection simultaneously projects onto the ``block-wise support'' of $X$. 

It is common and useful to rewrite the RIP inequalities such as in Definition~\ref{def:RIP} as an equivalent 
spectral condition of restrictions of $\mc A\ad \mc A$. 
\begin{proposition}
	Let $X \in \hat\Omega^{n,d}_{s,r}$ and $\mathcal A: \CC^{nd\times d} \to \RR^m$ a linear map. 
	Then the following two statements are equivalent: 
	\begin{enumerate}[a.)]
		\item $\snorm{\mc P_{\hat\Omega(X)}\circ (\Id - \mc A\ad \circ \mc A)\circ \mc P_{\hat\Omega(X)}} \leq \delta$.
		\item For all $Y \in \operatorname{range}{\mc P_{\hat\Omega(X)}}$ it holds that 
		\begin{equation}\label{prop:eq:fnormrip}
			(1-\delta)\fnorm{Y}^2 \leq \fnorm{\mc A(Y)}^2 \leq (1+\delta)\fnorm{Y}^2.
		\end{equation}
	\end{enumerate}
	\label{prop:spectralRIP}
\end{proposition}
\begin{proof}
	The inequality 
	\begin{align*}
		\delta &\geq \snorm{\mc P_{\hat\Omega(X)}\circ (\Id - \mc A\ad \circ \mc A)\circ \mc P_{\hat\Omega(X)}} \\
		&= \max_{Y \in \operatorname{range}{\mc P_{\hat\Omega(X)}}} \frac{
				|\langle Y, (\Id - \mc A\ad \circ \mc A) Y \rangle |
			}
			{
				\fnorm{Y}^2
			}
	\end{align*}
	holds if and only if for all $Y \in \operatorname{range}{\mc P_{\hat\Omega(X)}}$
	\begin{equation*}
		\delta \fnorm{Y}^2 \geq |\fnorm{Y}^2  - \fnorm{\mc A(Y)}^2|.
	\end{equation*}
	The last bound is equivalent to \eqref{prop:eq:fnormrip}.
\end{proof}

We will now prove the recovery guarantee, Theorem~\ref{thm:recGarant}. 
The derivation of recovery guarantees for the IHT algorithm follows largely the same blueprint 
developed in the original IHT proposal for sparse vectors \cite{BlumensathDavies:2008}, see 
also Ref.~\cite{FoucartRauhut:2013} for a detailed description of the proof.
Here, we are in addition in the comfortable position that Ref.~\cite{StrohmerWei:2017} already 
fleshed out the details of the recovery proof for an IHT algorithm for de-mixing low-rank matrices. 
However, in order to accommodate a non-trivial choice of the step width the proof of 
Ref.~\cite{StrohmerWei:2017} yields a slightly weaker result than what can be shown by a simpler argument for a fixed step width. 
Thus, we give a slightly simpler proof that carefully adapts the one given in Ref.~\cite{StrohmerWei:2017} to account for the additional sparsity constraint and uses a slightly more concise notation. 

\begin{proof}[Proof of Theorem~\ref{thm:recGarant}]
Let $X\in\hat{\Omega}_{s,r}^{n,d}$ be the matrix to be recovered. Let $X^l$ denote the $l$th iterate of the vector of matrices in the SDT algorithm (Algorithm~\ref{alg:SDTalgCompact}). Since the algorithm always involves a projection step onto $\hat{\Omega}_{s,r}^{n,d}$ the $l$th iterate $X^l$ is in $\hat{\Omega}_{s,r}^{n,d}$. Furthermore, we observe that $X+X^l+X^{l+1} \in \hat\Omega_{3s,3r}^{n,d}$. For convenience, we denote the projection onto the (``block-wise'') joint range and support of $X$, $X^l$ and $X^{l+1}$ simply by $\mathcal P^l \coloneqq \mathcal P_{\hat\Omega(X+X^l+X^{l+1})}$ and its orthogonal complement by $\mc P^l_\perp$. 
It is crucial for the proof to bound norm deviations restricted to the range of $\mathcal P^l$ as this 
eventually allows us to apply a RIP bound.

We want to show the convergence of the iterates of the algorithm $X^l$ to the correct solution $X$. 
In other words, we want to derive a bound of the form 
\begin{equation*}
	\fnorm{X^{l+1} - X} \leq \gamma \fnorm{X^{l} - X} 
\end{equation*}
with constant $\gamma < 1$. Note that by the theorem's assumption we set the step width to $\mu^l = 1$ 
and  omit the tangent space projection $P_{\mathcal T_{X^l}}$. 

We first derive the following consequence of the thresholding operation: 
Let $G^l \coloneqq \mc A\ad(y - \mc A(X^l)) = \mc A\ad\circ \mc A(X - X^l)$. By the definition of $X^{l+1}$ as the best approximation to $X^l + G^l$ in $\hat\Omega_{s,r}^{n,d}$ it holds that 
\begin{equation*}
\begin{split}
	\fnorm{X^{l+1} - \left[X^l + G^l \right]} 
	\leq \fnorm{X - \left[X^l + G^l \right]}. 
\end{split}
\end{equation*}
Since the parts of both sides of the inequality that are not in the kernel of $\mc P^l_\perp$ coincides, we get the same inequality also for the with $\mc P^l$ inserted
\begin{equation*}
\begin{split}
	&\fnorm{X^{l+1} - \left[X^l + \mc P^l(G^l) \right]} \\
	&\quad\leq \fnorm{X - \left[X^l + \mc P^l(G^l) \right]}.
\end{split} 
\end{equation*}
With the help of this inequality, we can bound 
\begin{equation}\label{eq:convProof:XminusXwithM}
	\begin{split}
		\fnorm{X^{l+1} - X} &\leq \fnorm{X^{l+1} - \left[X^l + \mc P^l(G^l) \right]} \\
		&\quad+ \fnorm{X - \left[X^l + \mc P^l(G^l) \right]} \\
		&\leq 2\fnorm{X - \left[X^l + \mc P^l(G^l) \right]} \\
		&= 2\fnorm{\mc M_{1} (X^l - X)}\\
	&\leq 2\norm{\mc M_{1}}_\infty \fnorm{X^l - X},
	\end{split}
\end{equation} 
where in the last step we used the definition of $G^l$, the fact that $\mathcal P^l$ acts trivially on $X^l - X$ and 
defined $\mc M_{1} \coloneqq \mc P^l\circ ( \Id - \mc A\ad \circ \mc A) \circ \mc P^l$. 
To arrive at the theorem's assertion, we now bound the spectral norm of $\mathcal M_{1}$ using the RIP property of $\mathcal A$ and Proposition~\ref{prop:spectralRIP}:
\begin{equation}\label{eq:convProof:M1}	
	\snorm{\mc M_{1}} = \snorm{\mc P^l\circ ( \Id - \mc A\ad \circ \mc A) \circ \mc P^l} \leq \delta_{3s,3r}
\end{equation}
since the range of $\mathcal P^l$ is in $\Omega^{n,d}_{3s,3r}$.
Using \eqref{eq:convProof:M1} in \eqref{eq:convProof:XminusXwithM} completes the proof.
\end{proof}
{Theorem~\ref{thm:recGarant} assumes that the input data $y$ for the SDT algorithm originate from a signal 
$X \in \hat\Omega_{s,r}$. In particular, we are assuming a bound on the block-sparsity $s$ and rank-$r$ of the blocks. In practice, one will often encounter the situation that the signal producing the data is not exactly sparse and of low-rank but rather well-approximated by a structured signal. 
It is straight-forward to also derive a model-robust version of Theorem~\ref{thm:recGarant}. 
We here briefly sketch the required modifications to the proof: 
Let $\tilde X \in \CC^{nd \times d}$ (an arbitrary signal without the hierarchical structure) and suppose that $y = \mathcal{A}(\tilde X) = \mathcal{A}(X) + \mathcal{A}(\tilde X - X)$ with $X = P_{\hat\Omega^{n,d}_{s,r}}(\tilde X)$ is the input to the SDT algorithm. 
In the proof, we can adapt the definition of $G^l$ to include the additional term involving the model-mismatch $\tilde X - X$.  
In the following steps, we can finally account for the model-mismatch term by an additional summand to the right-hand side of \eqref{eq:convProof:XminusXwithM} proportional to $\fnorm{\tilde X - X}$. In this way, we establish that $\lim_{l \to \infty} \fnorm{X^l - X} \leq C \fnorm{\tilde X - X}$ for some constant $C$, independent of $n$,$d$,$s$, and $r$.
We, thus, conclude that if the signal is violating the structure assumptions, the SDT algorithm converges to the projection of the signal onto the $\hat\Omega^{n,d}_{s,r}$ up to an accuracy proportional to the magnitude of the model-mismatch, measured as $\fnorm{\tilde X - P_{\hat\Omega^{n,d}_{s,r}}(\tilde X)}$.
}

\section{RIP guarantee for Hermitian random matrices}\label{app:RIPguarantees}
In this section we provide the proof of Theorem~\ref{thm:OmegaRIPGauss} that establishes 
the RIP condition for measurement matrices consisting of Hermitian matrices i.i.d.~drawn from the $\operatorname{GUE}$. 
Establishing RIP conditions for Gaussian matrices for a set of structured signals typically proceeds in two steps: 
One first derives a strong concentration result for a single signal in the set using standard concentration of measure. 
Second, one takes the union bound over the signal set with the help of an $\epsilon$-covering net construction to arrive at the uniform statement of RIP. We can readily adapt this strategy also to $\operatorname{GUE}$. 

For the first step, we derive a Gaussian-type concentration result, modifying a standard line of arguments for our example, 
see, e.g., Ref.~\cite{StrohmerWei:2017}. The result is summarized as the following lemma:
\begin{lemma}[Gaussian-type concentration]\label{lem:concentration}
	Let $X \in \hat\Omega^{n,d}_{s,r}$. Let $\{ {A}^{(k)}_{i}\}_{i=1, k=1}^{n,m}$ be a set of Hermitian matrices 
	drawn i.i.d.\ from the $\operatorname{GUE}$ and 
	$\mathcal{A}$ be the measurement operator defined by $\{ {A}^{(k)}_{i}\}_{i=1, k=1}^{n,m}$ via Eqs.~\eqref{eq:observed_data} and \eqref{eq:measurement_observable}. 
	Then, for $0 < \delta < 1$
	\begin{equation*}
		(1- \delta) \fnorm{X}^2 \leq \frac1m \lTwoNorm{\mathcal{A}(X)}^2 \leq (1 + \delta) \fnorm{X}^2 
	\end{equation*}
	with probability of at least $1 - 2 \e^{-m \delta^2 / C_\delta}$ and constant $C_\delta \geq 40$.
\end{lemma}
Our proof essentially follows the argument of Ref.~\cite{StrohmerWei:2017} for Gaussian measurements and then exploits that the Hermitian blocks of the signal $X \in \hat\Omega^{n,d}_{s,r}$ only overlap with the Hermitian part of the Gaussian measurement matrix. 

\begin{proof}
Let $\mat X \in \hat\Omega^{n,d}_{s,r}$ and denote its $n$ $d\times d$ blocks by $x_i$. 
Consider a set $ \{ B_i^{(k)} \in \CC^{d\times d} \}_{k,i=1}^{m, n}$ of $m\cdot n$ $d \times d$ matrices with entries independently drawn from the complex-valued normal distribution. 
Let $A_i^{(k)} \coloneqq P_\scrsym{} B_i^{(k)}$ be corresponding matrices drawn from the 
$\operatorname{GUE}$ and $\mathcal A$ the corresponding measurement map. 
Since all blocks $x_i$ are Hermitian, we have
\begin{equation*}
\begin{split}
	\mathcal{A}(X)^{(k)} &= 
		\sum_{i=1}^n \langle A_i^{(k)}, x_i \rangle 
		= \sum_{i=1}^n \langle P_\scrsym{} B_i^{(k)}, x_i \rangle \\
		&= \sum_{i=1}^n \Re\{\langle  B_i^{(k)}, x_i \rangle\}.
\end{split}
\end{equation*}
Since all entries of $B_i^{(k)}$ are i.i.d.\ complex normal random variables and $x_i$ is Hermitian, 
$\Re\{\langle  B_i^{(k)}, x_i \rangle\}$ are i.i.d.~real random variables 
from the distribution $\mathcal{N}(0,\|x_i\|^2_F)$ for all $i$ and $k$.
We conclude that all entries $y_k = \mathcal A(X)^{(k)}$ of $\mathcal A(X)$ are Gaussian distributed with variance $\sigma^2 = \sum_i\|x_i\|^2_F = \|X\|^2_F$ and have even moments $\EE[{y_k}^{2t}] = 2^{-t} t! \binom{2t}{t} \sigma^{2t}$ \cite[Corollary~7.7]{FoucartRauhut:2013}. 
Correspondingly, the squared entries are sub-exponential random variables with mean $\EE[y_k^2] = \sigma^2$. We denote the associated centred sub-exponential variable as 
\begin{equation*}
	z_k \coloneqq y_k^2 - \sigma^2.
\end{equation*}
The moments of $z_k$ are bounded by
\begin{equation*}
	\EE[ |z_k|^t] \leq 2^t \EE[ | y_k|^{2t} ] = t! \binom{2t}{t} \sigma^{2t},
\end{equation*}
where the first inequality follows from the triangle and Jensen's inequality. 
The binomial can be upper bounded using Stirling's formula \cite[(C.13)]{FoucartRauhut:2013} by $\binom{2t}{t} = 4^t r_t /\sqrt{\pi t}$ with $r_t \leq \e^{1/(24t)}$. Thus, for $t\geq 2$ we have $\EE[|z_k|^t] \leq t! R^{t-2} \Sigma^2 / 2$ with $R = 4\sigma^2$ and $\Sigma^2 = \sqrt{2/\pi} \e^{1/48} 16 \sigma^4 \leq 0.815 \cdot 16 \sigma^4$. 
Controlling the moments of $z_k$ for $t\geq 2$, we can apply the Bernstein inequality \cite[Theorem 7.30]{FoucartRauhut:2013} and bound the probability that $\|\mathcal{A}(\mat X)\|_{\ell_2}^2$ varies by more than $\Delta > 0$ from its expectation value
\begin{align}
\begin{split}
\label{eq:bernsteinbound}
  &\Pr \left[\left|\frac{1}{m}\left\|\mathcal{A}(\mat{X})\right\|_{\ell_2}^2- \|\mat{X}\|_F^2 \right|\geq \Delta\right] \\
  &\quad\quad=\Pr\left[ 
  	\left|\sum_{k=1}^m z_k\right| \geq m\Delta 
  \right]
\\
&\quad\quad\leq 2 
	\exp\left[
		-\frac{m\Delta^2 /2}
			{\Sigma^2 + R \Delta}
	\right]
\\
&\quad\quad\leq 2
	\exp\left[
		\frac{
			- m \Delta^2 
		}{
			32 \|X\|_F^4 + 8 \|X\|_F^2 \Delta
		}
	\right].
\end{split}
\end{align}
Let $ \Delta = \delta \| \mat{X} \|_F^2 $ for some $ 0< \delta <1 $. Then we can rewrite the tail bound~\eqref{eq:bernsteinbound} as
\begin{equation}\begin{split}
\label{equ:ProbRIP}
&\PP\left[\left| \frac1m \left\|\mathcal{A}(\mat{X})\right\|_{\ell_2}^2-\|\mat{X}\|_F^2 \right|\ge \delta \|\mat{X}\|_F^2\right] \\
&\qquad\leq 2 \exp \left[- \frac{m\delta^2}{C_\delta}\right] 
\end{split}\end{equation}
with a constant $C_\delta \geq 40$.
Hence, the condition
\begin{align*}
(1-\delta) \|\mat{X}\|_F^2 \leq \frac1m \left\| \mathcal{A}(\mat{X})\right\|_{\ell_2}^2 \leq (1+\delta) \|\mat{X}\|_F^2
\end{align*}
holds with probability at least $1-2\e^{-m\delta^2/C_\delta}$.
\end{proof}

Note that by the homogeneity of the RIP condition it suffices to restrict ourselves to normalized elements of $\hat\Omega^{n,d}_{s,r}$ in the proof of Theorem~\ref{thm:OmegaRIPGauss}. 
In the following, we will therefore focus on the set
\begin{equation*}
  \bar{\Omega}^{n,d}_{s,r} \coloneqq \{ \mat X \in \hat\Omega^{n,d}_{s,r} \mid \|\mat X\|_F^2 = 1 \}.
\end{equation*}
To take a union bound over the set $\bar\Omega^{n,d}_{s,r}$ we need to bound the size of an 
$\epsilon$-net that covers the set $\bar\Omega^{n,d}_{s,r}$. 
An $\epsilon$-net $\mathcal S$ covering a set of matrices $\mathcal M \subset \CC^{nd \times d}$ is a 
finite subset of $\mathcal M$ such that for all $\mat X \in \mathcal M$ there exists 
$\bar{\mat X} \in \mathcal S$ such that $\| \mat X - \bar{\mat X} \|_F \leq \epsilon$. Our construction generalizes the construction of Ref.~\cite{StrohmerWei:2017}. 
Therein, a covering net for the set of normalized block-wise low-rank matrices $\bar\Omega^{n,d}_{n,r}$ was derived. 
We summarize the statement given in Ref.~\cite{StrohmerWei:2017} in the following lemma without giving a proof. 
\begin{lemma}[Covering $\bar\Omega^{n,d}_{n,r}$ \cite{StrohmerWei:2017}]
\label{lem:strohmercovering}
  For $\bar\Omega^{n,d}_{n,r}$ there exists an $\epsilon$-covering net $\mathcal S^{n,d}_r$ with cardinality bounded by $(9/\epsilon)^{(2d+1)nr}$.
\end{lemma}

The proof of Lemma~\ref{lem:strohmercovering} basically lifts the result of an $\epsilon$-net for low-rank matrices of Ref.~\cite{CandesPlan:2011} to the set $\bar\Omega^{n,d}_{n,r}$ using the triangle inequality. \\

We can combine multiple $\epsilon$-nets for $\bar\Omega^{s,d}_{s,r}$ to construct an $\epsilon$-covering net for the set $\bar\Omega^{n,d}_{s,r}$ of block-sparse matrix vectors with low-rank blocks. 
The bound on the cardinality of the resulting $\epsilon$-covering net is given in the following lemma: 

\begin{lemma}[Bound on the cardinality of a covering net]\label{lem:covering}
  For $\bar\Omega^{n,d}_{s,r}$ there exists an $\epsilon$-covering net $\mathcal S^{n,d}_{s,r}$ of cardinality bounded by $\binom{n}{s}(9/\epsilon)^{(2d+1)sr}$. 
  Furthermore, for each $\mat X = [ \mat X_1, \ldots , \mat X_n] \in \bar\Omega^{n,d}_{s,r}$ there exists $\bar{\mat X} = [ \bar{\mat X}_1, \ldots, \bar{\mat X}_n] \in \mathcal S^{n,d}_{s,r}$ such that $\| \mat X - \bar{\mat  X}\|_F \leq \epsilon$ and $\|\bar{\mat X}_k\|_F = 0$ for all $k$ for which $\|\mat X_k\|_F = 0$.
\end{lemma} 

\begin{proof} 
Let $\Gamma \subset [n]$ with $|\Gamma| \leq s$, i.e.,\ the indices of the support of an $s$-sparse vector. The set
  \begin{equation*}
    \bar\Omega^\Gamma_r \coloneqq \left\{
      \sum_{i \in \Gamma} \xi_i e_i \otimes x_i 
      \ \middle| \ 
      \xi_i \in \RR,\ 
      x_i \in \mathcal D^d_r
    \right\} \subset \bar\Omega^{n,d}_{s,r}
  \end{equation*}
  shall consist of all elements of $\bar\Omega^{n,d}_{s,r}$ which have non-vanishing blocks only supported on $\Gamma$. 
  To each element of $\bar\Omega^\Gamma_r$, 
  we can associate an element of $\bar\Omega^{s,d}_{s,r}$ 
  by omitting the vanishing blocks in the matrix vector and vice versa. 
  By virtue of Lemma~\ref{lem:strohmercovering} we thus know that  $\bar\Omega^\Gamma_r$ has a covering net $\mathcal S^\Gamma_r$ of cardinality bounded by $(9/\epsilon)^{(2d+1)sr}$. 

  We can decompose the entire set $\bar\Omega^{n,d}_{s,r}$ as 
  \begin{equation*}
    \bar\Omega^{n,d}_{s,r} = \bigcup_{\Gamma \subset [n], |\Gamma|\leq s} \bar\Omega^\Gamma_r,
  \end{equation*}
  and thus, the set 
  \begin{equation*}
     \mathcal S^{n,d}_{s,r} = \bigcup_{\Gamma \subset [n], |\Gamma|\leq s}  \mathcal S^\Gamma_r
   \end{equation*} 
   is an $\epsilon$-covering net for $\bar\Omega^{n,d}_{s,r}$. 
   The union is taken over $\binom ns$ different sets. Thus, the cardinality of $\mathcal S^{n,d}_{s,r}$ is upper bounded by $\binom{n}{s}(9/\epsilon)^{(2d+1)sr}$. 
   The second statement follows by construction. 
\end{proof}

We are now in the position to prove Theorem~\ref{thm:OmegaRIPGauss}. 

\begin{proof}[Proof of Theorem ~\ref{thm:OmegaRIPGauss}]
The proof proceeds in two steps. 
First, we prove the RIP for elements of the $\epsilon$-covering net $\mathcal S^{n,d}_{s,r}$ of $\bar \Omega^{n,d}_{s,r}$. 
To do so, we combine the concentration result of Lemma~\ref{lem:concentration} and the 
union bound of Lemma~\ref{lem:covering} to establish uniform concentration. 
In a second step, following Ref.~\cite{StrohmerWei:2017}, 
we then use the definition of an $\epsilon$-covering net to show that for elements $\mat X \in \bar\Omega^{n,d}_{s,r}$ that are close enough to an element of the net, the RIP condition still holds.

\emph{Step 1:} Taking the union bound over the $\epsilon$-net $\mathcal{S}^{n,d}_{s,r}$ constructed in Lemma~\ref{lem:covering} and using the result of Lemma~\ref{lem:concentration} in the form of \eqref{equ:ProbRIP} with constant $C_\delta \geq 40$ we get
\begin{equation}
\label{equ:unionBound}
\begin{split}
	&\Pr\left( \max_{\substack{\mat{X}\in \mathcal{S}^{n,d}_{s,r}}}\left|\frac{1}{m}\|\mathcal{A}(\mat{X})\|_{\ell_2}^2-\|\mat{X}\|_F^2 \right| 
	\quad\geq \delta/2 \right) \\
	&\quad\leq 2| \mathcal{S}^{n,d}_{s,r}| \e^{-m\delta^2/(4C_\delta)} \\
	&\quad\leq 2\binom{n}{s}\left(\frac{9}{\epsilon}\right)^{(2d+1)sr} e^{-m\delta^2/(4C_\delta)}.
 \end{split}
\end{equation}
The aim is to find a lower bound for the number of measurements $m$ for which the probability \eqref{equ:unionBound} small. 
To this end, we rewrite
\begin{align}
\label{equ:calcm}
\begin{split}
	&2\binom{n}{s}\left(\frac{9}{\epsilon}\right)^{(2d+1)sr} e^{-m\delta^2/(4C_\delta)} \\
	&\quad\leq 2\exp\left[
		s\ln\frac{\e n} s
		+ (2d + 1) s r \ln\frac{9} \epsilon
		- \frac {m\delta^2}{4C_\delta}
	\right] \\
	&\quad\leq \tau ,
\end{split}
\end{align}
using 
$
\binom{n}{s}\leq \left( \frac{en}{s} \right)^s 
$
 \cite[Lemma~C.5]{FoucartRauhut:2013}.
The latter inequality becomes true under the condition that 
\begin{equation}\label{eq:mconditionNet}
	m \geq \frac{4 C_\delta}{\delta^2}
		\left[
			s \ln 
				\frac{\e n}
				{s}
			+ (2d + 1) sr \ln 
				\frac 9\epsilon 
			+ \ln
				\frac2\tau
		\right].
\end{equation}
Assuming that \eqref{eq:mconditionNet} holds, 
we have established the RIP condition for the $\epsilon$-net $\mathcal S^{n,d}_{s,r}$, i.e., for all vectors $\overline{\mat{X}} \in \mathcal S^{n,d}_{s,r}$ it holds that 
\begin{equation}
\label{equ:RIPforNet}
(1-\delta/2)\|\overline{\mat{X}}\|_F^2\leq\|\mathcal{A}(\overline{\mat{X}})\|_{\ell_2}^2 \leq(1+\delta/2)\|\overline{\mat{X}}\|_F^2
\end{equation}
with probability at least $1-\tau$. 

\emph{Step 2:} Let us now transfer the RIP of $\mathcal S^{n,d}_{s,r}$ to the entire set $\bar\Omega^{n,d}_{s,r}$ while keeping the error under control.
To this end, we choose the net parameter $\epsilon$ as $
	\frac{\delta}{4\sqrt{2}}
$. %
By definition of an $\epsilon$-net, for elements $\mat{X}\in\bar\Omega^{n,d}_{s,r}$, there exists an element $\overline{\mat{X}}\in\mathcal S^{n,d}_{s,r}$ such that 
\begin{equation}
\label{equ:epsilonCond}
\|\mat{X}-\overline{\mat{X}}\|_F\leq\frac{\delta}{4\sqrt{2}}.
\end{equation}
To prove the RIP for the set $\bar\Omega^{n,d}_{s,r}$ we need to bound $\|\mathcal{A}(\mat{X})\|_F$ from above and below. 

We start with the upper bound, making use of Eq.~\eqref{equ:RIPforNet}: 
\begin{equation}\label{equ:telescoping}
\begin{split}
  \|\mathcal{A}(\mat{X})\|_{\ell_2} 
&\leq\|\mathcal{A}(\overline{\mat{X}})\|_{\ell_2}+\|\mathcal{A}(\mat{X}-\overline{\mat{X}})\|_{\ell_2} \\
&\leq 1+\frac\delta 2+\|\mathcal{A}(\mat{X}-\overline{\mat{X}})\|_{\ell_2}.
\end{split}
\end{equation}
Now $\|\mathcal{A}(\mat{X}-\overline{\mat{X}})\|_{\ell_2}$ has to be bounded from above. 
We use that by the second statement of Lemma~\ref{lem:covering} the block supports of $\mat X$ and $\overline {\mat X}$ coincide. 
Therefore, $\mat{X}-\overline{\mat{X}}$ has also $s$ non-vanishing blocks that have rank of at most $2r$. 
We can, thus, decompose $\mat{X}-\overline{\mat{X}} = B + C$ in terms of orthogonal matrices $\mat B, \mat C \in \hat\Omega^{n,d}_{s,r}$ that obey $\langle\mat{B}, \mat{C}\rangle=0$.
In particular, $\mat B$ and $\mat C$ have the same block support as $\mat X$.
Let us define
\begin{equation*}
    \kappa_{s,r} \coloneqq \sup_{\mat X \in \bar\Omega^{n,d}_{s,r}} \| \mathcal A(\mat X) \|_{\ell_2}. 
\end{equation*}
Then we get 
using homogeneity 
\begin{equation*}
	\begin{split}
	\|\mathcal{A}&(\mat{X}-\overline{\mat{X}})\|_{\ell_2} 
	 \leq 
	\|\mathcal{A}(\mat{B})\|_{\ell_2} 
	+ \|\mathcal{A}(\mat{C})\|_{\ell_2} \\
	&\leq \kappa_{s,r} (\|\mat{B}\|_F + \| \mat C \|_F)
	\leq \sqrt{2}\, \kappa_{s,r} \sqrt{\|\mat B\|_F^2 + \| \mat C \|_F^2} \\
	&= \sqrt{2}\, \kappa_{s,r} \|\mat X - \bar{\mat X}\|_F,
  \end{split}
\end{equation*}
where the last step makes use of the orthogonality of $B$ and $C$.
Together with \eqref{equ:epsilonCond} it follows that
\begin{equation}\label{equ:differenceupperbound}
\|\mathcal{A}(\mat{X}-\overline{\mat{X}})\|_{\ell_2} \leq \frac{\delta\cdot \kappa_{s,r}}{4} . 
\end{equation}
It remains to derive an upper bound for $\kappa_{s,r}$. 
To this end, we use that, by definition, $\kappa_{s,r}$ is the best upper bound of the left-hand side of \eqref{equ:telescoping}. 
Inserting \eqref{equ:differenceupperbound} into the right-hand side of \eqref{equ:telescoping}, we find the condition
\begin{equation}\label{eq:kappaCondition}
   \kappa_{s,r} \leq 1 + \frac{\delta}{2} + \frac{\delta\cdot \kappa_{s,r}}{4}. 
\end{equation}
Solving for $\kappa_{s,r}$, Eq.~\eqref{eq:kappaCondition} implies for $0 < \delta < 1$ 
\begin{equation}
\label{eq:kappaBound}
\begin{split}
	\kappa_{s,r} 
		\leq 
			\frac{1 + \delta/2}{ 1 - \delta/4}
		\leq 1 + \delta. 
\end{split}\end{equation}
Altogether, this yields the desired upper bound
\begin{equation*}
    \|\mathcal{A}(\mat{X})\|_{\ell_2} 
    \leq 
    	 1 + \frac3 4 \delta + \frac{\delta^2} 4
    \leq 1 + \delta, 
\end{equation*}
for $\delta< 1$. 
The lower bound is analogously obtained by combining the inequality 
\begin{align*}
\|\mathcal{A} (\mat{X})\|_{\ell_2}  & \geq \|\mathcal{A}(\overline{\mat{X}})\|_{\ell_2}-\|\mathcal{A}(\mat{X}-\overline{\mat{X}})\|_{\ell_2}\\
& \geq 1-\delta/2-\|\mathcal{A}(\mat{X}-\overline{\mat{X}})\|_{\ell_2}
\end{align*}
with \eqref{equ:differenceupperbound}  
\eqref{eq:kappaBound} to arrive at
\begin{equation*}
\|\mathcal{A}(\mat{X})\|_{\ell_2} \geq  1-\delta/2-\delta(1+\delta)/4 \geq 1-\delta.
\end{equation*}
With the choice of $\epsilon$, we can rewrite the condition \eqref{eq:mconditionNet} on $m$ as 
\begin{equation*}
	m \geq \frac{C}{\delta^2}
		\left[ 
			s  \ln \frac{\e n} s
			+ (2d +1) sr \ln \frac c {\delta} 
			+ \ln\frac 2 \tau
		\right]
\end{equation*}
with constants $C \geq 4 C_\delta \geq 160$ and $c \geq 36 \sqrt 2 \geq 51 $.
This completes the proof. \\
\end{proof}

\tocless\section{Acknowledgements}
We thank David Gross, Steven T.~Flammia, Christian Krumnow, Robin Harper, Yi-Kai Liu, and Carlos A.~Riofrio 
for inspiring discussions and helpful comments.
Furthermore, we are grateful to Alireza Seif and Nobert Linke for valuable comments on realistic error models 
and Peter Jung for making us aware of Ref.~\cite{StrohmerWei:2017}. 
We are grateful to Susane Calegari for kindly providing drawings used in Figure~\ref{fig:viciouscycle}.
 This work has been supported by the DFG (specifically SPP1798 CoSIP, but also  EI 519/9-1, EI 519/7-1,
 CRC 183, as well as under Germany's Excellence Strategy - The Berlin Mathematics
Research Center MATH+, EXC-2046/1, project ID: 390685689),  the BMBF (DAQC, MUNIQC-ATOMS), 
and the Munich Quantum Valley (MQV-K8). 
It has also received funding from the Templeton Foundation and from the European Union's Horizon 2020 research and innovation programme 
 (PASQuanS2, Millenion).
D.~H.~acknowledges funding from the U.S.~Department of Defense through a QuICS Hartree fellowship. 
\vfill


\bibliographystyle{quantum}


\onecolumngrid
\end{document}